\pdfoutput=1 

\documentclass[11pt]{article}

\usepackage{comment}

\title{Approximating Tensor Network Contraction with Sketches}
\author{Mike Heddes\\University of California, Irvine\\\texttt{mheddes@uci.edu} \and Igor Nunes\\University of California, Irvine\\\texttt{igord@uci.edu} \and Tony Givargis\\University of California, Irvine\\\texttt{givargis@uci.edu} \and Alex Nicolau\\University of California, Irvine\\\texttt{nicolau@ics.uci.edu}}
\date{}

\usepackage[margin=1in]{geometry}
\usepackage[parfill]{parskip}

\usepackage[numbers]{natbib}

\usepackage{graphicx}
\usepackage{xcolor}
\usepackage{booktabs}
\usepackage[hidelinks]{hyperref}
\usepackage{pgf}
\usepackage{import}
\usepackage{tikz}
\usetikzlibrary{calc}
\usetikzlibrary{positioning}
\usetikzlibrary{shapes}
\usetikzlibrary{fit}
\usepackage{caption, float}
\usepackage{algorithm}
\usepackage{algpseudocode}

\algnewcommand{\LineComment}[1]{\State \(\triangleright\) #1}

\usepackage{enumitem}

\usepackage{amsthm}
\usepackage{thmtools}
\usepackage{thm-restate}

\usepackage{amsmath}
\usepackage{amsfonts}
\usepackage{amssymb}
\usepackage{amsthm}
\usepackage{mathtools}
\usepackage{bm}

\newcommand{\eqnote}[1]{&\text{\llap{#1}}}

\newtheorem{definition}{Definition}
\newtheorem{lemma}{Lemma}

\newtheorem{corollary}{Corollary}

\newtheorem{remark}{Remark}

\newtheorem{example}{Example}

\DeclareMathOperator{\E}{\mathbb{E}}          %
\DeclareMathOperator{\Var}{Var}               %
\DeclareMathOperator{\nnz}{nnz}               %
\DeclareMathOperator{\vecop}{vec}             %
\DeclareMathOperator{\mat}{mat}               %
\DeclareMathOperator{\tr}{tr}                 %
\DeclareMathOperator{\tc}{tc}                 %
\def\tran{^{\mathsf{T}}}                      %
\def\dft{F}                                   %
\def\idft{\dft^{-1}}                          %
\def\frob{\mathrm{F}}                         %

\DeclarePairedDelimiter\rdbr{\lparen}{\rparen}
\DeclarePairedDelimiter\sqbr{\lbrack}{\rbrack}
\DeclarePairedDelimiter\anbr{\langle}{\rangle}
\DeclarePairedDelimiter\set\{\}
\DeclarePairedDelimiter\abs{\lvert}{\rvert}

\DeclarePairedDelimiter\norm{\lVert}{\rVert}

\def\nats{\mathbb{N}}   %
\def\reals{\mathbb{R}}  %
\def\comps{\mathbb{C}}  %

\def\ivone{{\mathbf{1}}}

\def\ivi{{\mathbf{i}}}
\def\ivj{{\mathbf{j}}}

\def\ivl{{\mathbf{l}}}
\def\ivm{{\mathbf{m}}}
\def\ivn{{\mathbf{n}}}
\def\ivo{{\mathbf{o}}}

\def\ivr{{\mathbf{r}}}
\def\ivs{{\mathbf{s}}}

\def\vone{{\bm{1}}}

\def\vb{{b}}

\def\vy{{y}}

\def\evb{{b}}

\def\evy{{y}}

\def\mB{{B}}

\def\emB{{B}}

\newcommand{\tens}[1]{\mathcal{#1}}
\def\tA{{\tens{A}}}
\def\tB{{\tens{B}}}

\def\tX{{\tens{X}}}
\def\tY{{\tens{Y}}}

\newcommand{\etens}[1]{\mathcal{#1}}

\def\etA{{\etens{A}}}
\def\etB{{\etens{B}}}

\def\etX{{\etens{X}}}
\def\etY{{\etens{Y}}}

\definecolor{red}{HTML}{E31B23}
\definecolor{green}{HTML}{1BA13D}
\definecolor{blue}{HTML}{005CAB}
\definecolor{purple}{HTML}{78468E}
\definecolor{orange}{HTML}{EA8810}

\newcommand{\RNum}[1]{\uppercase\expandafter{\romannumeral #1\relax}}

\begin{document}

\maketitle

\begin{abstract}

Tensor network contraction is a fundamental mathematical operation that generalizes the dot product and matrix multiplication. It finds applications in numerous domains, such as database systems, graph theory, machine learning, probability theory, and quantum mechanics. 
Tensor network contractions are computationally expensive, in general requiring exponential time and space.
Sketching methods include a number of dimensionality reduction techniques that are widely used in the design of approximation algorithms.
The existing sketching methods for tensor network contraction, however, only support acyclic tensor networks. 
We present the first method capable of approximating arbitrary tensor network contractions, including those of cyclic tensor networks.
Additionally, we show that the existing sketching methods require a computational complexity that grows exponentially with the number of contractions. 
We present a second method, for acyclic tensor networks, whose space and time complexity depends only polynomially on the number of contractions. 

\end{abstract}

\section{Introduction}
\label{sec:introduction}

Tensor network contraction (TNC) is a fundamental mathematical operation with applications in various disciplines, including quantum physics, machine learning, database systems, graph theory, and probability theory. Analogous to how tensors generalize vectors and matrices to higher orders, TNC generalizes the dot product and matrix multiplication. 
TNC can be understood as a \textit{tensor contraction} (see Definition~\ref{def:tensor-contraction}) applied to the \textit{tensor product} (see Definition~\ref{def:tensor-product}) of multiple tensors. 
A tensor contraction is specified by a set of shared indices that form connections between tensors, creating a \textit{tensor network}.

The definition of TNC below uses the \textit{Iverson bracket} notation $\sqbr{S}$, which equals to one if the statement $S$ is true and zero otherwise; the consecutive integer set $\sqbr{k} \equiv \set{1, \dots, k}$ for any non-negative integer $k$; and $\ivi_K$ to denote the \textit{multi-index} $\ivi \equiv (i_1, \dots, i_q)$, which is restricted to the indices in $K \subseteq \sqbr{q}$. More details on tensors and multi-index notation are provided in Section~\ref{sec:background}.

\begin{definition}[Tensor contraction]
\label{def:tensor-contraction}
    Given any order-$q$ tensor $\tX \in \reals^{n_1 \times \dots \times n_q}$ and contractions $E \subset \sqbr{q} \times \sqbr{q}$, such that for every $(u, v) \in E$, $u \neq v$ and $n_u = n_v$.
    Let $K = \set{u, v : (u, v) \in E}$ and $K' = \sqbr{q} \setminus K$. 
    The tensor contraction $\tc(\tX, E)$ is an order-$\abs{K'}$ tensor specified by
\begin{align*}
    \tc(\tX, E)(\ivi_{K'}) = \sum_{\ivi_{K} = \ivone_{K}}^{\ivn_{K}} \etX(\ivi) \prod_{(u,v) \in E} \sqbr{i_u = i_v}, \qquad \forall \ivone_{K'} \leq \ivi_{K'} \leq \ivn_{K'}.
\end{align*}
\end{definition}

\begin{definition}[Tensor product]
\label{def:tensor-product}
Given any two tensors $\tA \in \reals^{n_1 \times \cdots \times n_p}$ and $\tB \in \reals^{m_1 \times \cdots \times m_q}$ of orders $p$ and $q$ respectively, their tensor product $\tA \times \tB$ is specified by $(\tA \times \tB)(\ivi, \ivj) = \etA(\ivi)\etB(\ivj)$
for all $\ivone \leq \ivi \leq \ivn$ and $\ivone \leq \ivj \leq \ivm$, where $\tA \times \tB \in \reals^{n_1 \times \cdots \times n_p \times m_1 \times \cdots \times m_q}$.
\end{definition}

\begin{definition}[Tensor network contraction (TNC)]
\label{def:tensor-network-contraction}
For any $p \in \nats^{+}$, tensors $\tX_1, \dots, \tX_p$, and contractions $E$, the tensor network contraction is given by $\tc(\tX_1 \times \cdots \times \tX_p, E)$.
\end{definition}

The following example of a TNC is included to clarify the definitions above.
The tensor network in Example~\ref{example:tensor-network-contraction} is shown in Figure~\ref{fig:example-tensor-diagram} using \textit{tensor diagram notation}, also known as \textit{Penrose graphical notation}. In this notation, tensors are denoted by shapes with legs, where each leg represents a mode, and connected legs denote a contraction.

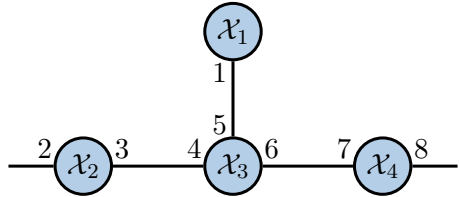
\begin{figure}[h]
    \centering
    \scalebox{1}{\begin{tikzpicture}[auto, node distance=18mm, thick, main node/.style={circle, fill=blue!30, draw, minimum size=0.75cm, inner sep=0pt}, every node/.style={line width=0.4mm}, every path/.style={line width=0.4mm}]
    \node[main node] (1) {$\tX_1$};
    \node[main node] (3) [below=10mm of 1] {$\tX_3$};
    \node[main node] (2) [left=12mm of 3] {$\tX_2$};
    \node[main node] (4) [right=12mm of 3] {$\tX_4$};
    \node[main node, fill=none, draw=none] (l2) [left=6mm of 2] {};
    \node[main node, fill=none, draw=none] (r4) [right=6mm of 4] {};

    \draw[-] (1) -- (3) node [pos=0.15, left=-2pt] {1} node [pos=0.85, left=-2pt] {5};
    \draw[-] (2) -- (3) node [pos=0.10, above=-1pt] {3} node [pos=0.90, above=-1pt] {4};
    \draw[-] (3) -- (4) node [pos=0.10, above=-1pt] {6} node [pos=0.90, above=-1pt] {7};
    \draw[-] (2) -- (l2) node [pos=0.20, above=-1pt] {2} node [pos=0.90, above=-1pt] {};
    \draw[-] (4) -- (r4) node [pos=0.20, above=-1pt] {8} node [pos=0.90, above=-1pt] {};
\end{tikzpicture}}
    \caption{Tensor diagram of an example tensor network}
    \label{fig:example-tensor-diagram}
\end{figure}

\begin{example}
\label{example:tensor-network-contraction}
The TNC $\tc(\tX_1 \times \tX_2 \times \tX_3 \times \tX_4, E) = \tY$ (shown in Figure~\ref{fig:example-tensor-diagram}), consists of the four tensors $\tX_1 \in \reals^{n_1}$, $\tX_2 \in \reals^{n_2 \times n_3}$, $\tX_3 \in \reals^{n_4 \times n_5 \times n_6}$, and $\tX_4 \in \reals^{n_7 \times n_8}$; and the three contractions $E = \set{(1, 5), (3, 4), (6, 7)}$.
Then, from Definition~\ref{def:tensor-contraction}, we have $K = \set{1, 3, 4, 5, 6, 7}$, $K' = \set{2, 8}$, and after a reduction of the Iverson brackets, we get
\begin{align*}
    \etY(i_2, i_8) = \sum_{i_1 = 1}^{n_1} \sum_{i_3 = 1}^{n_3} \sum_{i_6 = 1}^{n_6} \etX_1(i_1)\etX_2(i_2, i_3)\etX_3(i_3, i_1,i_6)\etX_4(i_6, i_8), \qquad \forall (1,1) \leq (i_2, i_8) \leq (n_2, n_8).
\end{align*}
\end{example}

The general problem of TNC is known to be \textsf{NP}-hard~\cite{bridgeman2017hand, arad2010quantum}. 
In \cite{markov2008simulating}, the authors show that exact algorithms have a running time exponential in the treewidth of the tensor network.
In this paper, we introduce a polynomial-complexity $(\epsilon, \delta)$-approximation for TNCs of acyclic tensor networks based on sketching. 
Additionally, we present the first method capable of approximating arbitrary TNCs, including those of cyclic tensor networks.
To our knowledge, this is the first work to directly address the general problem with sketching, and it establishes a new bound for important instances of the problem.

Sketching, or dimensionality reduction, has attracted significant attention for its ability to swiftly scale down problems while maintaining an approximation of the solution space.
This technique has been applied to a number of fundamental problems in numerical linear algebra and machine learning, including least squares regression and low-rank approximation \cite{woodruff2014sketching}.
Interest in these methods surged following the influential paper \cite{alon1996space}, which introduced the \textit{AMS sketch}---a method for estimating the second frequency moment of a data stream.
A notable improvement of the AMS sketch, called the \textit{count sketch} \cite{charikar2002finding}, significantly reduced the sketching time complexity.
Both sketches later gained popularity as inner product estimators \cite{alon1999tracking, cormode2005sketching}.

In the realm of database theory, the estimation of inner products holds significant importance. It serves as a core aspect of query optimization, in particular, join size estimation. This importance has motivated the database community to further advance sketching methods. 
Among the notable results are two generalizations of inner product sketches to accommodate multi-join queries: first of the AMS sketch \cite{dobra2002processing}, and later of the count sketch \cite{heddes2024convolution}.
Building upon this, our work introduces a new perspective by showing that the problem of join size estimation, as addressed in these previous works, corresponds to the problem of approximating TNC.

Intuitively, expanding the estimation from the join size of single-joins to that of multi-joins, mirrors the broader generalization from approximating inner products to approximating TNCs. 
Given this perspective, the motivation for our work arises from the observation that (1) the aforementioned join size estimators only support acyclic queries and (2) their accuracy deteriorates exponentially with an increase in the number of joins, which correspond to contractions. In fact, we show an exponential lower bound in Section~\ref{sec:mativation} for the technique presented in~\cite{heddes2024convolution}.

While our initial motivation stems from database theory, our reformulation of the underlying problem into one of TNC reveals the broader applicability of our proposed technique. In Section~\ref{sec:applications}, we discuss the important role that the more general problem of approximating TNC plays beyond databases. 
For example, TNCs are used to simulate quantum computers \cite{bridgeman2017hand}; in probability theory, TNC corresponds to marginalization in graphical models \cite{robeva2019duality}; and counting the number of triangles in a graph can be reduced to a TNC.

\subsection{Technical overview}

We present two methods that are $(\epsilon, \delta)$-approximate tensor network contractions (ATNCs) (see Definition~\ref{def:approx-tensor-contraction}).
The error bound of an ATNC is defined relative to the norm of the input tensor, rather than the output of the TNC. 
This is unavoidable as any error bound relative to the output of a TNC would imply a sublinear-space one-way protocol for the \emph{Index} problem, contradicting its $\Omega(n)$ randomized communication lower bound \cite{jayram2008one}.
A similar input-relative error bound is standard for approximate matrix multiplication \cite{woodruff2014sketching}, which is a special case of a TNC\footnote{For matrices $A, B \in \reals^{n \times n}$, the matrix multiplication $AB = \tc(A \times B, \set{(2, 3)})$.}.

\begin{definition}[Approximate tensor network contraction (ATNC)]
    \label{def:approx-tensor-contraction}
    Given $\epsilon, \delta > 0$ and any positive integer $p$, a random function $\hat{\tc}$ is an $(\epsilon, \delta)$-approximate tensor network contraction when, for any tensor network contraction $\tc(\tX, E)$ with $\tX = \tX_1 \times \cdots \times \tX_p$, 
    \begin{align*}
        \Pr\rdbr*{\norm*{\hat{\tc}(\tX, E) - \tc(\tX, E)}_{\frob} \leq \epsilon \norm{\tX}_{\frob}} \geq 1 - \delta.
    \end{align*}
\end{definition}

Our first method enables the approximation of \textit{any} TNC, while prior sketching methods were only able to approximate \textit{acyclic} TNCs \cite{deeds2025data}.
A TNC is \textit{cyclic} when its tensor network contains a cycle, and is \textit{acyclic} otherwise.
Key to this result is the \textit{complement count sketch}, a circularly reversed count sketch, introduced in Section~\ref{sec:general-method}. Each contraction is assigned an independent count sketch and its complement so that the estimation process can use circular convolution of the sketches instead of circular cross-correlation. We show that the use of circular cross-correlation in \cite{heddes2024convolution} prevents it from approximating cyclic TNCs.
We then show an exponential lower bound for the current best methods in Section~\ref{sec:mativation}, motivating the need for a new approach. 

Our second method eliminates the exponential dependence on the number of contractions for acyclic TNCs. To accomplish this, we interpret the tensor network as a tree structure and formulate acyclic TNC as a recursive function, consisting of a series of matrix multiplications with Kronecker products, in Section~\ref{sec:acyclic-method}. The Kronecker products are then sketched using the recursive sketching technique from \cite{ahle2020oblivious}. 
To ensure a space and time-efficient algorithm, the sketches are constructed incrementally from the leaves of the tree up to the root. 
These two methods yield our main result, as stated in Theorem~\ref{thm:errorbound}. We state our main result for \textit{full} TNCs, which are TNCs that result in a scalar value. In Section~\ref{sec:main_result}, we further demonstrate that our methods can also be applied to \textit{partial} TNCs, which result in a tensor of nonzero order.

\begin{restatable}{theorem}{errorbound}
\label{thm:errorbound}
For every $p \in \nats^{+}$, any order-$q_k$ tensors $\tX_k$ for $k \in \sqbr{p}$, and every $\epsilon, \delta > 0$, there exists an $(\epsilon, \delta)$-ATNC of a full TNC that can be computed in time $O((pm \log m + q N) \log 1/\delta)$ using $O(m p \log 1/\delta)$ space, with {\bf (1)} $m = \Omega(t/\epsilon^2)$ in the acyclic setting and {\bf (2)} $m = \Omega(3^t/\epsilon^2)$ in general, where $N = \sum_{k=1}^p \nnz(\tX_k)$, $q = \sum_{k=1}^p q_k$, and $t = \abs{E}$.
\end{restatable}

In Table~\ref{tab:complexity-approx-tensor-contraction}, we compare the time and space complexities of the two proposed methods with the existing approaches. We parameterize the complexity with $N = \sum_{k=1}^p \nnz(\tX_k)$, $q = \sum_{k=1}^p q_k$, and $t = \abs{E}$, which denote the number of non-zero components, modes, and contractions, respectively. Note that the existing sketching methods only support acyclic tensor networks. In this setting, our method brings an exponential improvement in the sketching dimension and, hence, in time and space complexity. Additionally, we present the first sketching method that can approximate general TNCs, including those with cyclic tensor networks.

\begin{table}[h]
    \centering
    \caption{Complexity comparison of approximate full tensor network contraction}
    \label{tab:complexity-approx-tensor-contraction}
    \begin{tabular}{l|cccc}
    \toprule
        Method & Setting & Time & Space & Sketch size ($m$) \\
    \midrule
        \cite{dobra2002processing} & Acyclic & $O((pm + m q N) \log 1/\delta)$ & $O(m p \log 1 / \delta)$ & $\Omega(3^t / \epsilon^2)$ \\
        \cite{heddes2024convolution} & Acyclic & $O((pm \log m + q N)\log 1 / \delta)$ & $O(m p \log 1 / \delta)$ & $\Omega(3^t / \epsilon^2)$ \\
        \textit{Ours} & Acyclic & $O((pm \log m + q N) \log 1 / \delta)$ & $O(m p \log 1 / \delta)$ & $\Omega(t / \epsilon^2)$\\
        \textit{Ours} & General & $O((pm \log m + q N) \log 1 / \delta)$ & $O(m p \log 1 / \delta)$ & $\Omega(3^t / \epsilon^2)$\\
    \bottomrule
    \end{tabular}
\end{table}

\section{Preliminaries}
\label{sec:background}

In this section, we introduce the notation and concepts used throughout the paper. 
We first revisit concepts from tensor algebra and their notation, based on \cite{golub2013matrix}. 

An \textit{order-$q$} tensor $\tA \in\reals^{n_1 \times \cdots \times n_q}$ can be conceptualized as a $q$-dimensional array where the index of the $k$th \textit{mode} ranges from $1$ to $n_k$. A component $\etA(i_1, \dots, i_q)$ of an order-$q$ tensor is specified by $q$ indices, with $1 \leq i_k \leq n_k$, for all $k$ from 1 to $q$. 
With \textit{multi-index} $\ivi = (i_1, \dots, i_q)$, a tensor component is equivalently indexed as $\etA(\ivi) \equiv \etA(i_1, \dots, i_q)$. The all-ones multi-index of length $q$ is denoted by $\ivone_q$, or simply $\ivone$ when its length is clear from the context. 
If $\ivi$ and $\ivj$ are multi-indices of the same length, then $\ivi \leq \ivj$ means that $i_k \leq j_k$ for all $k$.

The squared \textit{Frobenius norm} of $\tA$ is given by $\norm{\tA}_{\frob}^2 = \sum_{\ivi = \ivone}^{\ivn} \etA(\ivi)^2 \equiv \sum_{i_1 = 1}^{n_1} \cdots \sum_{i_q = 1}^{n_q} \etA(i_1, \dots, i_q)^2$,
the number of nonzero components of $\tA$ is denoted by $\nnz(\tA)$,
and the \textit{Hadamard product} is the component-wise multiplication between tensors.

\begin{definition}[Hadamard product]        
\label{def:hadamard-product}
Given any two tensors $\tA, \tB \in \reals^{n_1 \times \cdots \times n_q}$, their Hadamard product $\tA \circ \tB$ is specified by $(\tA \circ \tB)(\ivi) = \etA(\ivi)\etB(\ivi)$ for $\ivone \leq \ivi \leq \ivn$,
where $\tA \circ \tB \in \reals^{n_1 \times \cdots \times n_q}$.
\end{definition}

Besides indexing components of a tensor, it is also possible to extract \textit{subtensors}. A first and second-order extraction from a tensor are called a \textit{fiber} and a \textit{slice}, respectively. For example, if $\tA \in \reals^{2 \times 3 \times 4 \times 5}$ and $\vb = \etA(1, :, 2, 4)$, then $\vb$ is a fiber of $\tA$ with $\evb(i) = \etA(1, i, 2, 4)$. Moreover, if $B = \etA(1, :, 2, :)$, then $\mB$ is a slice of $\tA$ with $\emB(i,j) = \etA(1, i, 2, j)$. In this notation, a colon is used to indicate the extraction of a mode. 

A tensor can be reshaped in various ways. The \textit{vectorization} of a tensor $\tA \in\reals^{n_1 \times \cdots \times n_q}$, denoted as $\vecop(\tA) \in\reals^{n_1 \cdots n_q}$, is a column vector with the components of $\tA$ in lexicographic order. The \textit{mode-$k$ flattening} of $\tA$, given by $\mat_k(\tA) \in \reals^{n_k \times n_1 \cdots n_q / n_k}$, is a matrix whose columns correspond to the mode-$k$ fibers of $\tA$ arranged in lexicographic order. When $k$ is omitted, as in $\mat(\tA)$, it is assumed to be one. 

We continue our discussion with operations specific to first and second-order tensors, better known as vectors and matrices.
We start by defining the \textit{Kronecker product} between two matrices.

\begin{definition}[Kronecker product]
\label{def:kronecker-product}
Given any two matrices $A \in \reals^{m \times n}$ and $B \in \reals^{p\times q}$, their \textit{Kronecker product} $A \otimes B$ is an $(mp)$-by-$(nq)$ matrix with components given by $(A \otimes B)(\alpha, \beta) = A(i, j)B(k, l)$, where $\alpha=p(i-1)+k$ and $\beta=q(j-1)+l$.
\end{definition}

When the Kronecker product is applied to vectors, the vectors are assumed to be column matrices. This means that for vectors $x$ and $y$, $x \otimes y = \vecop(x \times y)$. With $x^{\otimes p} \equiv x \otimes \cdots \otimes x$ we denote the $p$-times Kronecker product of the vector $x$ with itself. A sequence of Kronecker products for every member $u$ of the set $S$ is denoted by $\bigotimes_{u \in S} X_u$. For an empty set, we have that $\bigotimes_{u \in \emptyset} X_u = 1$. 
The \textit{row-wise Kronecker product} applies the Kronecker product to the rows of two matrices. 

\begin{definition}[Row-wise Kronecker product]
\label{def:row-wise-kronecker-product}
Given any two matrices $A \in \reals^{p \times n}$ and $B \in \reals^{p\times m}$, their \textit{row-wise Kronecker product} $A \bullet B$ is a $p$-by-$(nm)$ matrix with components given by $(A \bullet B)(k, \alpha) = A(k, i)B(k, j)$, where $\alpha=m(i-1)+j$.
\end{definition}

For vectors $x, y \in \comps^n$, their inner product is denoted by $\anbr{x, y}$, the component-wise complex conjugate of $x$ is denoted by $\overline{x}$, and $e_i$ is the $i$th standard basis vector.
With $x * y$, we denote their \textit{circular convolution}, and with $x \star y$ their \textit{circular cross-correlation}. 

\begin{definition}[Circular convolution]
\label{def:circ-conv}
Given two vectors $x, y \in \comps^n$, their circular convolution $x * y$ is a vector with components given by $(x * y)_j = \sum_{i=1}^{n} x_i y_{(j-i \bmod n) + 1}$.
\end{definition}

\begin{definition}[Circular cross-correlation]
\label{def:circ-cross}
Given two vectors $x, y \in \comps^n$, their circular cross-correlation $x \star y$ is a vector with components given by $(x \star y)_j = \sum_{i=1}^{n} \overline{x}_i y_{(j+i -2 \bmod n) + 1}$.
\end{definition}

With $I$ we denote the identity matrix, and the matrix $F_n \in \comps^{n \times n}$, or simply $F$ when $n$ is clear from the context, denotes the \textit{discrete Fourier transform} matrix. The products $\dft_n x$ and $\idft_n x$ can be efficiently computed in time $O(n \log n)$ using the \textit{fast Fourier transform}. 
The following lemma relates circular convolution and circular cross-correlation with the Hadamard product.

\begin{lemma}
\label{lemma:conv-thm}
For any vectors $x, y \in \comps^{n}$, $\dft(x * y) = \dft x \circ \dft y$, and $\dft(x \star y) = \overline{\dft x} \circ \dft y$.
\end{lemma}

\subsection{Count sketch}
\label{sec:count-sketch}

The \textit{count sketch} (see Definition~\ref{def:count-sketch}) was introduced to find frequent items in data streams \citep{charikar2002finding}. It has since been used more broadly as a dimensionality reduction technique \cite{woodruff2014sketching}. Each column of the count sketch matrix has a single random sign in a random row. The sign is determined by a 4-wise independent hash function (see Definition~\ref{def:kwisehash}) and the row is determined by a 2-wise independent hash function. The count sketch of a vector $x$ can be efficiently computed in time $O(\nnz(x))$ if $x$ is represented as a list of its nonzero components.

\begin{definition}[$k$-wise independence]
\label{def:kwisehash}
A family of hash functions $H = \set{h: \sqbr{n} \to \sqbr{m}}$ is said to be $k$-wise independent if for any $k$ distinct items $x_1, \dots, x_k \in \sqbr{n}$ the hashed values $h(x_1), \dots, h(x_k)$ are independent and uniformly distributed in $\sqbr{m}$.
\end{definition}

\begin{definition}[Count sketch]
\label{def:count-sketch}
The count sketch matrix is a random matrix $C \in \reals^{m \times n}$, with each component specified as $C(j,i) = s(i)\sqbr{h(i) = j}$, where $s:\sqbr{n}\to\set{-1, 1}$ and $h:\sqbr{n}\to\sqbr{m}$ are drawn from families of 4 and 2-wise independent hash functions, respectively. Given any vector $x \in \mathbb{R}^n$, the count sketch of $x$ is obtained by $Cx$.
\end{definition}

\begin{lemma}[{\citep[Lemma 2]{weinberger2009feature}}]
\label{lemma:count-sketch}
For any positive integers $m, n$, let $C \in \reals^{m \times n}$ be a count sketch matrix. Then, for any vectors $x, y \in \reals^n$,
\begin{align*}
\E\sqbr*{\rdbr{C x}\tran \rdbr{C y}} = x\tran y \quad\text{and}\quad\Var\rdbr*{\rdbr{C x}\tran \rdbr{C y}} \leq \frac{2}{m} \norm{x}^2_2\norm{y}^2_2.
\end{align*}
\end{lemma}

\subsection{Tensor sketch}
\label{sec:tensor-sketch}

The \textit{tensor sketch} (see Definition~\ref{def:tensor-sketch}) was a pioneer in introducing a version of the count sketch that can be used to efficiently sketch tensor products. Its insights came forth from work on compressed matrix multiplication \citep{pagh2013compressed} applied to approximate the polynomial kernel $\anbr{x, y}^q = \anbr{x^{\otimes q}, y^{\otimes q}}$ for any vectors $x$ and $y$, in order to accelerate kernel machines \citep{pham2013fast}. 

\begin{definition}[Tensor sketch]
\label{def:tensor-sketch}
The tensor sketch matrix $T \in \reals^{m \times n^q}$ of order $q$ is a random matrix defined as $T = \idft((\dft C_1) \bullet (\dft C_2) \bullet \cdots \bullet (\dft C_q))$,
where $C_1, \dots, C_q \in \reals^{m \times n}$ are independent count sketch matrices. Given any vector $x \in \mathbb{R}^{n^q}$, the tensor sketch of $x$ is obtained by $Tx$.
\end{definition}

With tensor sketch matrix $T \in \reals^{m \times n^q}$ and a vector $x \in \mathbb{R}^{n^q}$, the sketch $Tx$ can be computed in time $O(q \nnz(x))$. However, given a vector $y \in \mathbb{R}^n$, the sketch $T(y^{\otimes q})$ can be computed in time $O(qm \log m + q \nnz(y))$ with the fast Fourier transform. This can be significantly more efficient than first computing $y^{\otimes q}$ and then applying the tensor sketch.

\begin{lemma}[{\citep[Lemma 2]{avron2014subspace}}]
\label{lemma:tensor-sketch}
    For any positive integers $m, n, q$, let $T \in \reals^{m \times n^q}$ be a tensor sketch matrix of order $q$. Then, for any vectors $x, y \in \reals^{n^q}$,
    \begin{align*}
    \E\sqbr*{\rdbr{T x}\tran \rdbr{T y}} = x\tran y \quad\text{and}\quad\Var\rdbr*{\rdbr{T x}\tran \rdbr{T y}} \leq \frac{3^q}{m}\norm{x}^2_2\norm{y}^2_2.
    \end{align*}
\end{lemma}

\subsection{Recursive sketch}
\label{sec:recursive-sketch}

The \textit{recursive sketch} (see Definition~\ref{def:recursive-sketch}) improves on a limitation of the tensor sketch, whose variance has an exponential dependence on the order of the tensor. The main insight is that the exponential dependence can be avoided when the modes of a tensor are sketched recursively, following a binary tree. Multiple sketches are presented in \cite{ahle2020oblivious} which can be used to perform the recursion. We use a recursion of tensor sketches in our definition below. 
For a given sketch size, the recursive sketch only increases the time complexity of the tensor sketch by a constant factor, while offering an exponential improvement in approximation accuracy.

\begin{definition}[Recursive sketch]
\label{def:recursive-sketch} 
The recursive sketch matrix $R \in \reals^{m \times n^q}$ of order $q$, with $q$ a power of two\footnote{
For notational simplicity, we allow $q$ to be any positive integer. The actual order is implied to be the next power of two $\tilde{q}$, and the tensors in $\reals^{n^q}$ are embedded into $\reals^{n^{\tilde q}}$ through Kronecker products with $e_1$, as described in \cite{ahle2020oblivious}. Importantly, this preserves the variance bound and does not affect the asymptotic complexity.
}, is a random matrix defined as $R = Q_2 Q_4 \cdots Q_{q/2} Q_q S$, where $S = C_1 \otimes \cdots \otimes C_q$ and $Q_l = T_{l,1} \otimes T_{l,2} \otimes \cdots \otimes T_{l,l/2}$, with independent count sketch matrices $C_1, \dots, C_q \in \reals^{m \times n}$, and independent tensor sketch matrices $T_{l,1}, \dots, T_{l, l/2} \in \reals^{m \times m^2}$ for $l \in \set{2, 4, \dots, q/2, q}$.
Given any vector $x \in \reals^{n^q}$, the recursive sketch of $x$ is obtained by $Rx$.
\end{definition}

\begin{lemma}[{\citep[Lemmas 12 and 15]{ahle2020oblivious}}]
\label{lemma:recursive-sketch}
    For any positive integers $m, n, q$, let $R \in \reals^{m \times n^q}$ be a recursive sketch matrix of order $q$. Then, for any vectors $x, y \in \reals^{n^q}$,
    \begin{align*}
    \E\sqbr*{\rdbr{R x}\tran \rdbr{R y}} = x\tran y \quad\text{and}\quad\Var\rdbr*{\rdbr{R x}\tran \rdbr{R y}} \leq \rdbr*{\!\rdbr*{1 + \frac{8}{m}}^{2q} - 1}\norm{x}^2_2\norm{y}^2_2.
    \end{align*}
\end{lemma}

\section{Related work}
\label{sec:related-work}

In this section, we review related works and compare them with our contributions.

\textbf{Sketch-based tensor network contraction.}
A generalization of the AMS sketch to enable the join size estimation of multi-join queries was proposed in \cite{dobra2002processing}. Although not explicitly mentioned in \cite{dobra2002processing}, the problem of join size estimation is equivalent to the problem of approximating a full TNC (see Section~\ref{sec:card-est}).
Their approach involves sketching an order-$q$ tensor using $q$ independent AMS sketch matrices. 
Once all the tensors are sketched, the approximation is obtained by performing the Hadamard product of the sketches, followed by an averaging of their components. 
Since the AMS sketch matrix is dense, computing the $m$-dimensional sketch of an order-$q$ tensor $\tX$ requires $O(mq \nnz(\tX))$ time. 
The method proposed in \cite{heddes2024convolution} improves this by a factor $m$ by using sparse count sketch matrices to sketch the tensors.
Both methods only support acyclic tensor networks and have a space and time complexity that depends exponentially on the number of contractions (see Table~\ref{tab:complexity-approx-tensor-contraction} for a comparison with our methods).

\textbf{Marginalization.}
Robeva and Seigal~\cite{robeva2019duality} established that marginalization in undirected graphical models with discrete variables is equivalent to TNC. Consequently, classical inference algorithms for graphical models can be viewed as contraction algorithms for tensor networks. This includes exact methods such as belief propagation on acyclic tensor networks and the junction tree algorithm for cyclic tensor networks~\cite{wainwright2008graphical}. As shown in \cite{wainwright2008graphical} and \cite{markov2008simulating}, the running time of such exact algorithms is exponential in the treewidth of the tensor network. To mitigate this complexity, approximate methods such as loopy belief propagation~\cite{wainwright2008graphical} are often employed in practice, though they typically lack worst-case approximation guarantees.

\textbf{Tensor decomposition.}
Tensor decomposition methods aim to factor a given tensor into a tensor network, with the specific network structure determined by the chosen decomposition (e.g., CP, Tucker, or tensor train). This differs from the TNC problem studied in this paper, where the tensor network is given and the goal is to efficiently compute its contraction. A substantial body of work studies the use of randomized algorithms and sketching to accelerate tensor decompositions. For instance, \cite{wang2015fast} proposes algorithms for CP decomposition that use sketches to speed up tensor power iterations. In \cite{malik2018low} and \cite{ma2021fast} sketched alternating least squares methods were presented for low-rank Tucker decomposition. More recently, \cite{mahankali2024near} presented algorithms for tensor train, Tucker, and CP decompositions, as well as decompositions into arbitrary tensor networks. While these works demonstrate the effectiveness of sketching for tensor decomposition, they do not address the problem studied in this paper.

\section{Methods}
\label{sec:method}

We propose two methods that are $(\epsilon, \delta)$-approximate tensor network contractions. We first describe our methods for full TNCs and later show that our methods also approximate partial TNCs, which result in a tensor of nonzero order.

After establishing four simplifying assumptions (without loss of generality), we start this section by describing our first method, which can approximate any TNC. 
We then demonstrate the need for a new sketching technique by showing that the variance of the current best method requires an exponential dependence on the number of contractions. Motivated by this, we present a recursive expression for acyclic full TNC that serves as a foundation for the introduction of our second method, which has a polynomial dependence on the number of contractions.

\subsection{Without loss of generality assumptions}
\label{sec:wlog}

To simplify our explanations and notation going forward, we make the following four assumptions about TNC, without loss of generality (WLOG):

\noindent{\bf WLOG 1.} Every mode is part of exactly one contraction. Otherwise, \textit{virtual copies} of the modes participating in multiple contractions can be created. Consider the following TNC of vectors $a, b, c$: $\tc(a \times b \times c, \set{(1, 2), (2, 3)})$, where the single mode of $b$ is part of two contractions. Defining $B(i, j) = b(i) \sqbr{i = j}$, this can instead be written as $\tc(a \times B \times c, \set{(1, 2), (3, 4)})$.

\noindent{\bf WLOG 2.} Every contraction involves two tensors. When the modes being contracted belong to the same tensor, such as in the trace of a matrix, the exact contraction can trivially be computed in time proportional to the number of non-zero components of the tensor.

\noindent{\bf WLOG 3.} At most a single contraction exists between any two tensors. Multiple contractions can be merged into one larger contraction by reshaping the tensors. The following example of matrices $A$ and $B$: $\tc(A \times B, \set{(1, 3), (2, 4)})$ can be expressed as $\tc(\vecop(A) \times \vecop(B), \set{(1, 2)})$.

\noindent{\bf WLOG 4.} Every mode has the same number of components. The methods discussed in this paper depend on the number of nonzero components of the tensors; therefore, zeros can be appended to smaller tensors to ensure that all modes have the same number of components.

It is easy to verify that these transformations do not increase the number of nonzero components or the Frobenius norm of the transformed tensors. Additionally, they do not introduce any extra tensors or contractions. These are sufficient conditions for our results to hold.

\subsection{Approximating any tensor network contraction}
\label{sec:general-method}

The insights for our first method stem from an observation of the method in \cite{heddes2024convolution}; namely, that the use of circular cross-correlation to combine sketches results in a pattern where each contraction has one conjugated mode when the TNC is acyclic.
This pattern arises because circular cross-correlation conjugates its first argument when expressed as a Hadamard product (see Lemma~\ref{lemma:conv-thm}),
resulting in a layering of complex conjugates that cancel out when their number is even. 
The tensors in every other level of the contraction tree therefore have their modes conjugated, ensuring that one mode of every contraction is conjugated. However, this relies on the tensor network having a tree structure and therefore does not work for cyclic tensor networks.

\begin{figure}[h]
    \centering
    \scalebox{1}{\begin{tikzpicture}[thick, main node/.style={circle, fill=blue!30, draw, minimum size=0.75cm, inner sep=0pt}, every node/.style={line width=0.4mm}, every path/.style={line width=0.4mm}]
    
    \node[main node] (1) at (0, 0) {$x$};
    \node[main node] (2) at (2cm, 0) {$Y$};
    \node[main node] (3) at (4cm, 0) {$z$};
    
    \node[main node] (4) at (8cm, -0.866cm) {$X$};
    \node[main node] (5) at (10cm, -0.866cm) {$Z$};
    \node[main node] (6) at (9cm, 0.866cm) {$Y$};

    \draw[-] (1) -- (2) node [pos=0.10, above=-1pt] {1} node [pos=0.90, above=-1pt] {2};
    \draw[-] (2) -- (3) node [pos=0.10, above=-1pt] {3} node [pos=0.90, above=-1pt] {4};
    
    \draw[-] (4) -- (5) node [pos=0.10, below=-1pt] {1} node [pos=0.90, below=-1pt] {6};
    \draw[-] (5) -- (6) node [pos=0.25, right=-1pt] {5} node [pos=0.90, right=-1pt] {4};
    \draw[-] (4) -- (6) node [pos=0.25, left=-1pt] {2} node [pos=0.90, left=-1pt] {3};
    
\end{tikzpicture}}
    \caption{Tensor diagram of a simple acyclic (left) and cyclic (right) tensor network}
    \label{fig:acycle-v-cycle}
\end{figure}
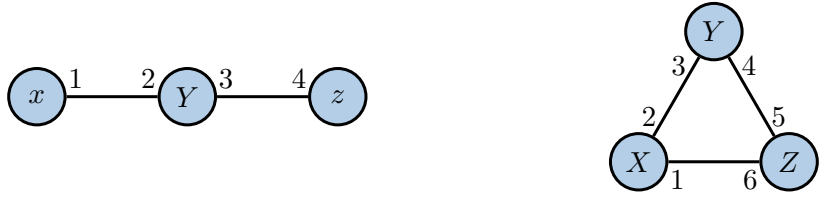

The following two examples of simple acyclic and cyclic TNCs (see Figure~\ref{fig:acycle-v-cycle}) show that each contraction in the acyclic TNC (Example~\ref{example:prior-method-acyclic}) has one conjugated mode (both are part of tensor $Y$), while contraction $(1, 6)$ of the cyclic TNC (Example~\ref{example:prior-method-cyclic}) has no conjugated mode. We emphasize this difference because our proof for the general setting (presented in Appendix~\ref{sec:proof-lemma-general-expectation-variance}) relies on each contraction having one conjugated mode.

\begin{example}
    \label{example:prior-method-acyclic}
    With the method in \cite{heddes2024convolution}, the acyclic TNC $\tc(x \times Y \times z, \set{(1, 2), (3, 4)})$ (shown on the left of Figure~\ref{fig:acycle-v-cycle}) with $x, z \in \reals^n$ and $Y \in \reals^{n \times n}$ is approximated as follows. First, contractions $(1, 2)$ and $(3, 4)$ are assigned independent count sketch matrices $C_1$ and $C_2$, respectively. The count sketches $C_1x$ and $C_2z$ and tensor sketch $\idft(\dft C_1 \bullet \dft C_2)\vecop(Y)$ are then computed and combined using circular cross-correlation. Finally, the first component of the resulting vector is selected, yielding the following expression:
\begin{align*}
    e_1\tran (C_1x \star \idft(\dft C_1 \bullet \dft C_2)\vecop(Y) \star C_2z)
    = e_1\tran \idft (\dft C_1x \circ \overline{(\dft C_1 \bullet \dft C_2)\vecop(Y)} \circ \dft C_2z ).
\end{align*}
\end{example}

\begin{example}
    \label{example:prior-method-cyclic}
    With the method in \cite{heddes2024convolution}, the cyclic TNC $\tc(X \times Y \times Z, \set{(2, 3), (4, 5), (1, 6)})$ (shown on the right of Figure~\ref{fig:acycle-v-cycle}) with $X, Y, Z \in \reals^{n \times n}$ would be approximated as follows. First, contractions $(2, 3)$, $(4, 5)$, and $(1, 6)$ are assigned independent count sketch matrices $C_1$, $C_2$ and $C_3$, respectively. The tensor sketches $\idft(\dft C_1 \bullet \dft C_2)\vecop(X)$, $\idft(\dft C_2 \bullet \dft C_3)\vecop(Y)$, and $\idft(\dft C_3 \bullet \dft C_1)\vecop(Z)$ are then computed and combined using circular cross-correlation. Finally, the first component of the resulting vector is selected, yielding the following expression:
\begin{align*}
    & e_1\tran (\idft(\dft C_1 \bullet \dft C_2)\vecop(X) \star \idft(\dft C_2 \bullet \dft C_3)\vecop(Y) \star \idft(\dft C_3 \bullet \dft C_1)\vecop(Z))\\
    &= e_1\tran \idft ((\dft C_1 \bullet \dft C_2)\vecop(X) \circ \overline{(\dft C_2 \bullet \dft C_3)\vecop(Y)} \circ (\dft C_3 \bullet \dft C_1)\vecop(Z) ).
\end{align*}
\end{example}

To describe our method for estimating cyclic TNCs, we first introduce the \textit{complement count sketch}, which is key to our approach.

\begin{definition}[Complement count sketch]
\label{def:complement-count-sketch}
    Given a count sketch matrix $C \in \reals^{m \times n}$ with hash functions $s : \sqbr{n} \to \set{-1, 1}$ and $h: \sqbr{n} \to \sqbr{m}$, the complement count sketch matrix $C'$ of $C$ is specified by $C'(j, i) = s(i) \sqbr{- h(i) + 2 \equiv j \,(\mathrm{mod}\, m)}$. 
\end{definition}

\begin{remark}
The complement of a count sketch matrix is itself a count sketch matrix, and the complement of a complement yields the original, that is, $(C')' = C$.
\end{remark}

The complement count sketch is the circularly reversed count sketch, which implies that $\dft C' x = \overline{\dft C x}$ for any $x \in \reals^n$ as shown in Lemma~\ref{lemma:complement-count-sketch} (with proof in Appendix~\ref{sec-proof-lemma-complement-count-sketch}). 
Instead of relying on circular cross-correlation to conjugate one mode of every contraction, the complement count sketch matrix can be used to explicitly control which mode of each contraction is conjugated. 

\begin{lemma}
\label{lemma:complement-count-sketch}
        For any positive integers $m$ and $n$, let $C \in \reals^{m \times n}$ be a count sketch matrix and $C'$ its complement. Then, for any $x \in \reals^n$, $\dft C' x = \overline{\dft C x}$.
\end{lemma}

Our method first samples independent count sketch matrices for one mode of each contraction and then assigns its complement count sketch matrix to the other mode. Specifically, for each $(u, v) \in E$, $C_u$ is an independent count sketch matrix and $C_v \coloneq C'_u$. The approximation of any full TNC is then given by
\begin{align}
    \hat{\tc}(\tX_1 \times \cdots \times \tX_p, E) = e_1 \tran \idft (x_1 \circ \cdots \circ x_p), \qquad x_k = \rdbr*{\bm{\bullet}_{u \in V(k)} \dft C_u}\vecop(\tX_k), \label{eq:ours-general}
\end{align}
where $V(k)$ are the modes of $\tX_k$ and $\bm{\bullet}_{u \in S}$ denotes the row-wise Kronecker product for each member $u$ of $S$. 
The following example shows that, with our method, each contraction of the simple cyclic TNC used in Example~\ref{example:prior-method-cyclic} now has one conjugated mode. The proof of the expectation and variance bound of our estimator in Lemma~\ref{lemma:general-expectation-variance} (with proof in Appendix~\ref{sec:proof-lemma-general-expectation-variance}) relies on this property.

\begin{example}
\label{example:method-1}
    With our first method, the cyclic TNC $\tc(X \times Y \times Z, \set{(2, 3), (4, 5), (1, 6)})$ (shown on the right of Figure~\ref{fig:acycle-v-cycle}) with $X, Y, Z \in \reals^{n \times n}$ is approximated as follows. First, modes 2, 4, and 1 are assigned independent count sketch matrices $C_2, C_4$, and $C_1$, respectively. Then, modes 3, 5, and 6 are assigned the complement count sketch matrices $C_2', C_4'$, and $C_1'$, respectively. The tensor sketches $\idft(\dft C_1 \bullet \dft C_2)\vecop(X)$, $\idft(\dft C'_2 \bullet \dft C_4)\vecop(Y)$, and $\idft(\dft C'_4 \bullet \dft C'_1)\vecop(Z)$ are then computed and combined using circular convolution. Finally, the first component of the resulting vector is selected, yielding the following expression:
\begin{align*}
    & e_1\tran (\idft(\dft C_1 \bullet \dft C_2)\vecop(X) * \idft(\dft C'_2 \bullet \dft C_4)\vecop(Y) * \idft(\dft C'_4 \bullet \dft C'_1)\vecop(Z) )\\
    &= e_1\tran \idft ((\dft C_1 \bullet \dft C_2)\vecop(X) \circ (\overline{\dft C_2} \bullet \dft C_4)\vecop(Y) \circ (\overline{\dft C_4} \bullet \overline{\dft C_1})\vecop(Z) ).
\end{align*}
\end{example}

We now show that our first method is an unbiased estimator for arbitrary TNCs, and bound its variance through a reduction to the tensor sketch in Lemma~\ref{lemma:general-expectation-variance} (with proof in Appendix~\ref{sec:proof-lemma-general-expectation-variance}). While our first method has the same variance bound and asymptotic complexity as the method in \cite{heddes2024convolution}, to the best of our knowledge, our method is the first sketching method that enables the estimation of cyclic TNCs.

\begin{lemma}
\label{lemma:general-expectation-variance}
    For every $m, p \in \nats^{+}$ and any full TNC, let $\hat{\tc}$ be defined as in Eq. {\normalfont (\ref{eq:ours-general})} with sketch size $m$. 
    Then, with $y = \hat{\tc}(\tX_1 \times \cdots \times \tX_p, E)$ and $t = \abs{E}$,
\begin{align*}
    \E\sqbr{y} = \tc(\tX_1 \times \cdots \times \tX_p, E), \qquad \text{and} \qquad \Var\rdbr*{y} \leq \frac{3^t}{m}\prod_{k=1}^p \norm{\tX_k}_{\frob}^2.
\end{align*}
\end{lemma}

The pseudocode for our first method is provided in Algorithm~\ref{alg:general-ftc}. Our estimator supports the most general streaming data setting, the general turnstile, in which a stream of updates can change any component by any amount. This is realized by storing the sketches $x_k$ for all $k \in \sqbr{p}$ (see Algorithm~\ref{alg:general-ftc} line \ref{alg-line:sketch}), with each update being applied in time $O(q_k)$.

\begin{algorithm}
\caption{Approximating general full tensor network contractions.}
\label{alg:general-ftc}
\begin{algorithmic}[1]
\Require tensors $\tX_k \in \reals^{n \times \cdots \times n}$ of degree $q_k$ for $k \in \sqbr{p}$, contractions $E$, sketch size $m$
\Ensure approximate tensor network contraction $y \in \reals$
\For{$(u, v) \in E$}
\State Let $C_u \in \reals^{m \times n}$ be an independent count sketch matrix.
\State $C_v \gets C'_u$
\EndFor
\State $z \in \set{1}^m$
\For{$k \in \sqbr{p}$}
\State $T_k \gets \idft(\bm{\bullet}_{u \in V(k)} \dft C_u)$ \Comment{Compose a  tensor sketch matrix of count sketch matrices}
\State $x_k \gets T_k \vecop(\tX_k)$  \label{alg-line:sketch}
\State $z \gets z \circ \dft x_k$
\EndFor
\State $y \gets \frac{1}{m} \sum_{i = 1}^{m} z_i$
\end{algorithmic}
\end{algorithm}

\subsection{Exponential lower bound}
\label{sec:mativation}

The best known bounds for approximate full TNC using sketching methods, including those introduced in \cite{dobra2002processing},  \cite{heddes2024convolution}, and our first method in Section~\ref{sec:general-method}, all have an exponential dependence on the number of contractions, as shown in Table~\ref{tab:complexity-approx-tensor-contraction}. 
To motivate the need for a new approach, in Lemma~\ref{lemma:exponential-dependence} (with proof in Appendix~\ref{sec:proof-lemma-exponential-dependence}) we claim that the method in \cite{heddes2024convolution} requires an exponential dependence on the number of contractions.
In particular, we show a lower bound on the variance for approximating a sequence of matrix multiplications based on \citep[Lemma 43]{ahle2020oblivious}, which provides a variance lower bound for norm approximation using the tensor sketch.

\begin{lemma}
    \label{lemma:exponential-dependence}
    For every positive integers $n, m, q$, where $q$ is even, let $C_1, \dots, C_q \in \reals^{m \times n}$ be independent count sketch matrices, and $T_i = \idft((\dft C_{i-1}) \bullet (\dft C_{i}))$.
    Then, for the all-ones tensors $\tX_1, \tX_{q+1} \in \set{1}^n$ and $\tX_i \in \set{1}^{n \times n}$ for $2 \leq i \leq q$, 
    \begin{align*}
    \Var\rdbr*{e_1\tran(C_1 \tX_1 \star T_2\vecop(\tX_2) \star \cdots \star T_q \vecop(\tX_q) \star C_q \tX_{q+1})} \geq \rdbr*{\frac{3^q}{2m^2} - 1} \prod_{i=1}^{q+1} \norm{\tX_i}_{\frob}^2.
    \end{align*}
\end{lemma}

This result, which extends to our first method, demonstrates that the exponential dependence of the prior best methods cannot be overcome merely through a tighter analysis. We thereby conclude that a new method is necessary for achieving an improved bound for the problem of ATNC. In the remainder of this section, we introduce our second approximation method and show how this new approach eliminates the exponential dependence on the number of contractions.

\subsection{Improved approximation for acyclic tensor network contraction}
\label{sec:acyclic-method}

For the second method, besides assuming a full TNC, we assume that the tensor network is acyclic. 
Acyclic tensor networks are common in several important settings. In particular, many relational database queries have an acyclic join graph, and estimating the join size of these queries (which corresponds to a TNC) is critical for efficient database systems (see Section~\ref{sec:card-est}). 
Although tractable algorithms exist for TNC with acyclic tensor networks, the intermediate tensors may still be prohibitively large, motivating the need for approximate TNC.

To facilitate the explanation of our proposed solution, we first introduce some additional notation and concepts. Specifically, we introduce the notion of a root tensor $\tX_o$, with $o \in \sqbr{p}$, which is arbitrarily selected from the contracted tensors to form a rooted tree structure from the tensor network.
We then define $\Gamma(k)$ as the outgoing neighbors\footnote{For notational simplicity, the neighbor function $\Gamma$ is assumed to yield neighbors in an order matching the modes of the corresponding tensor, and the first mode of a non-root tensor is assumed to contract with its parent in the tree. This arrangement can always be achieved by simply permuting the indices of the tensor.} of tensor $\tX_k$, that is, the set of tensors that share a contraction with $\tX_k$ and farther from the root tensor. 
Note that all the leaves of the rooted tree are vectors.
With this, we are able to write acyclic full TNC recursively as shown in Lemma~\ref{lemma:recursive-AFTNC} (with proof in Appendix~\ref{sec:proof-lemma-recursive-AFTNC}).

\begin{lemma}
\label{lemma:recursive-AFTNC}
    Given any $o \in \sqbr{p}$, any acyclic full TNC can be expressed recursively as 
    \begin{align*}  
    \tc(\tX_1 \times \cdots \times \tX_p, E) = \vecop(\tX_o) \tran g(o), \qquad g(k) = \bigotimes_{l \in \Gamma(k)} \mat(\tX_l)g(l).
    \end{align*}
\end{lemma}

\label{sec:sketch-construction}

Our second approach for approximate TNC builds upon the following insight: the recursive formulation of TNC consists of a sequence of matrix multiplications involving Kronecker products. Leveraging this observation, we sketch the Kronecker products using recursive sketch matrices. This leads to the following approximation of TNCs:
\begin{align}
    \hat{\tc}(\tX_1 \times \cdots \times \tX_p, E) = \vecop(\tX_o)\tran R_o\tran R_o \hat{g}(o), \qquad \hat{g}(k) = \bigotimes_{l \in \Gamma(k)} \mat(\tX_l)R_l\tran R_l \hat{g}(l), \label{eq:ours-acyclic}
\end{align}
where each $R_k$ maps a vectorized order-$(q_k - 1)$ tensor to a vector of size $m$, with the exception of $R_o$, which maps a vectorized order-$q_o$ tensor to a vector of size $m$.

In Lemma~\ref{lemma:acyclic-expectation-variance}, we show that our second method is an unbiased estimator of acyclic full TNC, and we bound its variance. The proof, provided in Appendix~\ref{sec:proof-lemma-acyclic-expectation-variance}, bounds the variance by accumulating $(1 + 8/m)^p$ terms from the recursive sketch variance bound in Lemma~\ref{lemma:recursive-sketch} while traversing the rooted tree. The expectation result follows from the expectation of the recursive sketch.

\begin{lemma}
\label{lemma:acyclic-expectation-variance}
For every $m, p \in \nats^{+}$ and any acyclic full TNC, let $\hat{\tc}$ be defined as in Eq. {\normalfont (\ref{eq:ours-acyclic})} with sketch size $m$. 
Then, with $y = \hat{\tc}(\tX_1 \times \cdots \times \tX_p, E)$ and $t = \abs{E}$,
    \begin{align*}
        \E\sqbr{y} = \tc(\tX_1 \times \cdots \times \tX_p, E) \qquad \text{and} \qquad \Var\rdbr*{y} \leq \rdbr*{\!\rdbr*{1+\frac{8}{m}}^{2t} - 1}\prod_{k=1}^p \norm{\tX_k}_{\frob}^2.
    \end{align*}
\end{lemma}

Computing an estimate efficiently is not straightforward, particularly since each $\mat(\tX_l)R_l\tran$ results in a (potentially large) $n$-by-$m$ matrix. %
To address this, we decompose the recursive sketch matrix from the previous recursion step and integrate its initial dimensionality reduction directly into the following step of the recursion. Since $R_k = Q_kS_k$, where $S_k = \bigotimes_{l \in \Gamma(k)} C_{k, l}$ is a Kronecker product of independent count sketch matrices, we can use the mixed-product identity $(A \otimes B)(x \otimes y) = (Ax) \otimes (By)$ to obtain the following factorization:
\begin{align*}
\vecop(\tX_o)\tran R_o\tran Q_oS_o \rdbr*{\bigotimes_{l \in \Gamma(o)} \mat(\tX_l)R_l\tran R_l \hat{g}(l)} = \vecop(\tX_o)\tran R_o\tran Q_o \rdbr*{\bigotimes_{l \in \Gamma(o)} C_{o,l}\mat(\tX_l)R_l\tran R_l \hat{g}(l)},
\end{align*}
which can be applied to every recursion step. This yields $C_{k,l}\mat(\tX_l)R_l\tran$, which are manageable $m$-by-$m$ matrices that can be efficiently computed in time $O(q_l \nnz(\tX_l))$. As a result, we have transformed a problem involving tensors of any order into a significantly more tractable problem that deals only with second-order tensors.

To eliminate the need for storing intermediate matrices altogether, we propose a progressive computation scheme that traverses the rooted tree of the tensor network from its leaves to the root. This approach efficiently computes the matrix-vector product during the creation of each subsequent sketch, resulting in compact sketch vectors at every recursion step.

Specifically, at the leaves we compute the count sketches of first-order tensors as $x_l = C_{k, l}\vecop(\tX_l)$, which takes $O(\nnz(\tX_l))$ time. Using these leaf-level sketch vectors, we then proceed one level up the tree. At this level, we first compute $r_k = Q_k \rdbr{\bigotimes_{l \in \Gamma(k)} x_l}$, which represents a recursive sketch of the leaf-level sketch vectors. We then compute $x_k = C_{l, k} \mat(\tX_k) R_k\tran r_k$, combining the computation of the sketch matrix $C_{l,k} \mat(\tX_k) R_k\tran$ and the matrix-vector product involving $r_k$. When carried out jointly, these operations require $O(q_k \nnz(\tX_k))$ time, improving upon the $O(q_k \nnz(\tX_k) + m^2)$ time needed if performed separately. 
The pseudocode for this procedure (called the \textsc{SketchedMatrixVectorProduct}) is provided in Appendix~\ref{sec:additional-algorithms}, Algorithm~\ref{alg:sketch-matrix-vector-product}, and 
the pseudocode for our second method is provided in Algorithm~\ref{alg:acylic-ftc}.
Example~\ref{example:method-2} illustrates our second method applied to the acyclic tensor network shown in Figure~\ref{fig:example-acyclic-tensor-diagram}.

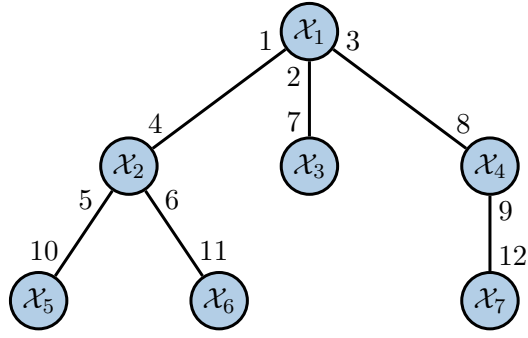
\begin{figure}[h]
    \centering
    \scalebox{1}{\begin{tikzpicture}[auto, node distance=10mm, thick, main node/.style={circle, fill=blue!30, draw, minimum size=0.75cm, inner sep=0pt}, every node/.style={line width=0.4mm}, every path/.style={line width=0.4mm}]

\node[main node] (x1) {$\tX_1$};

\node[main node] (x2) [below=of x1, xshift=-24mm] {$\tX_2$};
\node[main node] (x3) [below=of x1] {$\tX_3$};
\node[main node] (x4) [below=of x1, xshift=24mm] {$\tX_4$};

\node[main node] (x5) [below=of x2, xshift=-12mm] {$\tX_5$};
\node[main node] (x6) [below=of x2, xshift=12mm] {$\tX_6$};
\node[main node] (x7) [below=of x4] {$\tX_7$};

\draw (x1) -- (x2) node [pos=0.15, above=1pt] {1} node [pos=0.98, above=1pt] {4};
\draw (x1) -- (x3) node [pos=0.2, left=-1pt] {2} node [pos=0.8, left=-1pt] {7};
\draw (x1) -- (x4) node [pos=0.15, above=1pt] {3} node [pos=0.98, above=1pt] {8};

\draw (x2) -- (x5) node [pos=0.10, left=1pt] {5} node [pos=0.7, left=1pt] {10};
\draw (x2) -- (x6) node [pos=0.10, right=1pt] {6} node [pos=0.7, right=1pt] {11};
\draw (x4) -- (x7) node [pos=0.20, right=-1pt] {9} node [pos=0.8, right=-1pt] {12};

\end{tikzpicture}}
    \caption{Tensor diagram of an example acyclic tensor network}
    \label{fig:example-acyclic-tensor-diagram}
\end{figure}

\begin{example}
    \label{example:method-2}
    Using our second method, the TNC shown in Figure~\ref{fig:example-acyclic-tensor-diagram},
    \begin{align*}
        \tc(\tX_1 \times \tX_2 \times \tX_3 \times \tX_4 \times \tX_5 \times \tX_6 \times \tX_7, \set{(1, 4), (2, 7), (3, 8), (5, 10), (6, 11), (9, 12)}),
    \end{align*}
    with $\tX_1, \tX_2 \in \reals^{n \times n \times n}$, $\tX_4 \in \reals^{n \times n}$, and $\tX_3, \tX_5, \tX_6, \tX_7 \in \reals^n$ is approximated as follows. First, let $\tX_1$ be selected as the arbitrary root tensor, which allows us to interpret the tensor network as a rooted tree. Then, each tensor is assigned an independent recursive sketch matrix whose order matches the number of children the tensor has in the rooted tree. Thus, we have $R_1 \in \reals^{m \times n^3}$, $R_2 \in \reals^{m \times n^2}$, and $R_4 \in \reals^{m \times n}$. Recall that each recursive sketch matrix can be decomposed as $R_k = Q_k \rdbr{\bigotimes_{l \in \Gamma(k)} C_{k,l}}$, where $C_{k,l}$ are independent count sketch matrices. Now, starting at the lowest level of the tree, we compute the count sketches $x_5 = C_{2,5}\vecop(\tX_5)$, $x_6 = C_{2, 6}\vecop(\tX_6)$, $x_7 = C_{4, 7}\vecop(\tX_7)$. We then move one level up the tree, first computing the recursive sketches $r_2 = Q_2(x_5 \otimes x_6)$ and $r_4 = Q_4x_7$, and then $x_2 = C_{1, 2}\mat(\tX_2)R_2\tran r_2$, $x_3 = C_{1, 3}\vecop(\tX_3)$, and $x_4 = C_{1, 4}\mat(\tX_4)R_4\tran r_4$ using Algorithm~\ref{alg:sketch-matrix-vector-product}. Lastly, at the root of the tree, we compute $x_1 = R_1 \vecop(\tX_1)$, $r_1 = Q_1(x_2 \otimes x_3 \otimes x_4)$, and $y = x_1\tran r_1$.
\end{example}

\begin{algorithm}[h]
\caption{Approximating acyclic full tensor network contraction.}
\label{alg:acylic-ftc}
\begin{algorithmic}[1]
\Require tensors $\tX_k \in \reals^{n \times \cdots \times n}$ of degree $q_k$ for $k \in \sqbr{p}$, contractions $E$, sketch size $m$
\Ensure approximate tensor network contraction $y \in \reals$
\State Let $o$ be any member of $\sqbr{p}$.
\State Let $R_o \in \reals^{m \times n^{q_o}}$ be an independent recursive sketch.
\For{$k \in \sqbr{p} \setminus \set{o}$}
\State Let $R_k = Q_k(\bigotimes_{l \in \Gamma(k)} C_{k, l}) \in \reals^{m \times n^{q_k - 1}}$ be an independent recursive sketch.
\EndFor
\Function{NeighborsSketch}{$k$}
\State \textbf{if} {$\abs{\Gamma(k)} = 0$} \textbf{then} \Return $1$
\For{$l \in \Gamma(k)$}
\State $r_l \gets \textsc{NeighborsSketch}(l)$
\State $x_l \gets \textsc{SketchedMatrixVectorProduct}(C_{k, l}, \tX_l, R_l, r_l)$ \Comment{$C_{k, l} \mat(\tX_l) R_l\tran r_l$}
\EndFor
\State \Return $Q_k(\bigotimes_{l \in \Gamma(k)} x_l)$ \Comment{Apply remainder of the recursive sketch}
\EndFunction
\State $x_o \gets R_o\vecop(\tX_o)$
\State $r_o \gets \textsc{NeighborsSketch}(o)$
\State $y \gets x_o\tran r_o$
\end{algorithmic}
\end{algorithm}

\subsection{Main result}
\label{sec:main_result}

In this section, we combine the results from the previous sections to obtain our main result stated in Theorem~\ref{thm:errorbound}, which shows that our first method is an $(\epsilon, \delta)$-ATNC for any full TNC and that our second method is an $(\epsilon, \delta)$-ATNC for acyclic full TNCs.
The proof (presented in Appendix~\ref{sec:proof-thm-errorbound}) uses Chebyshev's inequality with the expected value and the variance upper bound from Lemmas~\ref{lemma:general-expectation-variance} and \ref{lemma:acyclic-expectation-variance} to obtain estimators with constant success probability.
By selecting the median of $O(\log 1 / \delta)$ independent estimates, each with a constant success probability, we can use the Chernoff bound to obtain an estimator that succeeds with probability $1 - \delta$. This approach is sometimes referred to as the \textit{median trick}. It is commonly applied to boost an estimator with constant success probability to success with high probability.

\errorbound*

Lastly, in Corollary~\ref{corollary:full-to-partial-tc} (with proof in Appendix~\ref{sec:proof-corollary-full-to-partial-tc}), we show that our two methods can also be applied to partial TNC, that is, those that evaluate to tensors of nonzero order. The time complexity of this generalization grows linearly with the number of components in the resulting tensor. Interestingly, the computational complexity of our methods only depends on the sizes of the input and output tensors.

\begin{corollary}
\label{corollary:full-to-partial-tc}
For every $p, n \in \nats^{+}$, any order-$q_k$ tensors $\tX_k \in \reals^{n \times \cdots \times n}$ for $k \in \sqbr{p}$, and every $\epsilon, \delta > 0$, there exists an $(\epsilon, \delta)$-ATNC of a partial TNC which results in an order-$\bar{q}$ tensor that can be computed in time $O((n^{\bar{q}} pm \log m + q N) \log 1/ \delta)$ using $O(n^{\bar{q}} m p \log 1/ \delta)$ space, with {\bf (1)} $m = \Omega(t/\epsilon^2)$ in the acyclic setting and {\bf (2)} $m = \Omega(3^t/\epsilon^2)$ in general, where $N = \sum_{k=1}^p \nnz(\tX_k)$, $q = \sum_{k=1}^p q_k$, and $t = \abs{E}$.
\end{corollary}

We have now established that our first method is an ATNC in general and that our second method is an ATNC in the acyclic setting. In the remainder of this manuscript, we highlight applications of TNC for which approximate solutions are permissible. 

\section{Applications of tensor network contraction}
\label{sec:applications}

TNC, by virtue of its generality, is important in many disciplines. From the mathematical modeling of physical phenomena to addressing both theoretical and practical issues in computer science. 
In quantum mechanics, tensors model many-body quantum wave functions, and approximate TNCs are used to simulate quantum computers \cite{bridgeman2017hand}. In machine learning, the contraction of efficient tensor representations seeks to mitigate the computational cost of training large models \cite{oseledets2011tensor,novikov2015tensorizing}. 
Meanwhile, in probability theory, the joint probability distribution of discrete random variables is naturally represented by a tensor, and marginalization in graphical models has been shown to be equivalent to TNC \cite{robeva2019duality}.
In graph theory, the number of $q$-colorings of a graph can be reduced to TNC \cite{bridgeman2017hand}.
In the remainder of this section, we highlight two applications of TNC: join size estimation in database systems and triangle counting in graphs.

\subsection{Join size estimation}
\label{sec:card-est}

One particularly relevant application of approximate TNC lies within database systems, namely join size estimation. In this setting, the join size is the number of tuples that result from a join. In general, the execution time of a query is highly sensitive to the join execution order \cite{leis2015good}. Since determining the optimal join order is known to be \textsf{NP}-hard \cite{ibaraki1984optimal}, query optimizers often rely on intermediate join size estimates\footnote{The intermediate query result can itself be expressed as a query (a subquery of the query to optimize).} as their primary input to assess the cost of a particular join order.

\begin{figure}[h]
    \centering
    \begin{minipage}{.45\textwidth}
\begin{verbatim}
SELECT COUNT(*) 
FROM R1, R2, R3, R4
WHERE R1.1 = R2.2 
AND R3.4 = R2.2 
AND R4.5 = R2.3
\end{verbatim}
\end{minipage}%
\begin{minipage}{0.45\textwidth}
    \begin{tikzpicture}[thick, main node/.style={circle, fill=blue!30, draw, minimum size=0.75cm, inner sep=0pt}, every node/.style={line width=0.4mm}, every path/.style={line width=0.4mm}]

    \node[main node] (A) at (120:3em) {$\tX_1$};
    \node[main node] (B) at (240:3em) {$\tX_3$};
    \node[main node] (D) at (0:3em) {$\tX_2$};
    \node[main node, right=3em of D] (E) {$\tX_4$};
    
    \draw[-] (A) -- (0, 0) node [pos=0.35, below=-1pt, left=-1pt] {1};
    \draw[-] (B) -- (0, 0) node [pos=0.35, above=-1pt, left=-1pt] {4};
    \draw[-] (D) -- (0, 0) node [pos=0.10, above=-1pt] {2};
    \draw[-] (D) -- (E) node [pos=0.10, above=-1pt] {3} node [pos=0.90, above=-1pt] {5};
    
\end{tikzpicture}
\end{minipage}%
    \caption{Example SQL query (left) and corresponding tensor network diagram (right). For all $k \in \sqbr{4}$, every component $\tX_k(\ivi)$ denotes the frequency of $\ivi$ in relation $R_k$.}
    \label{fig:example-graph-query}
\end{figure}
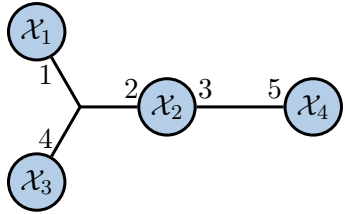

The join size of a query containing equi-joins can be expressed as a TNC, in which each relation is represented by a tensor and the joins specify the contractions \cite{alon1999tracking, dobra2002processing, heddes2024convolution}.
Join size estimation can therefore be formulated as an ATNC problem.
Consequently, our results yield improved bounds for join size estimation in the case of acyclic queries, and generalize previous results to cyclic queries as stated in Corollary~\ref{corollary:multi-join-card-est}.
An example SQL query and its corresponding tensor network diagram are provided in Figure~\ref{fig:example-graph-query}.

\begin{corollary}
\label{corollary:multi-join-card-est}
For every $p \in \nats^{+}$, any query of equi-joins of relations $R_k$ for $k \in \sqbr{p}$, and every $\epsilon, \delta > 0$, there exists an $(\epsilon, \delta)$-approximation of the join size (as in Definition~\ref{def:approx-tensor-contraction}) that can be computed in time $O((pm \log m + q N) \log 1/\delta)$ using $O(m p  \log 1/\delta)$ space, with {\bf (1)} $m = \Omega(t/\epsilon^2))$ when the query is acyclic and {\bf (2)} $m = \Omega(3^t/\epsilon^2)$ in general, where $N = \sum_{k=1}^p \abs{R_k}$, $q = \sum_{k=1}^p q_k$, and $t = \abs{E}$.
\end{corollary}

\begin{proof}
    Let $\tX_k$ be the frequency tensor of relation $R_k$, such that for all $k \in \sqbr{p}$, every component $\tX_k(\ivi)$ denotes the frequency of $\ivi$ in relation $R_k$, that is, $\tX_k(\ivi) = \sum_{\ivj \in R_k} \sqbr{\ivi = \ivj}$; and let $E$ be the set of equi-joins. Since the sketches are linear, the sketch of each frequency tensor can be computed in one pass over the relation. The results then follow directly from Theorem~\ref{thm:errorbound}.
\end{proof}

\subsection{Triangle counting}

Even for combinatorial problems for which tractable algorithms are known, ATNC can prove to be useful. Take, for instance, the task of counting triangles in graphs.
Given the emergence of applications involving massive graphs, even polynomial-time algorithms can become prohibitively expensive. Since the triangle count can be formulated as $\tr(AAA)$, with adjacency matrix $A$, it can also be approximated using ATNC. Efficient sketching-based algorithms for counting triangles have previously been explored, for example, in \cite{jowhari2005new}. In fact, the AMS sketch has inspired several such efficient algorithms for graph problems \cite{henzinger1998computing,feigenbaum2005graph,feigenbaum2004graph,feigenbaum2005graphSODA}. Our result for triangle counting presented in Corollary~\ref{corollary:triangle-count} matches the space bound of the second estimator in~\cite{jowhari2005new}, but improves upon their time complexity of $O(mN \log 1/\delta)$. Moreover, the second estimator in \cite{jowhari2005new} needs 12-wise independent hash functions, whereas our method relies on at most 4-wise independent hash functions, which are significantly more efficient in practice.

\begin{corollary}
\label{corollary:triangle-count}
For any directed graph $G$ with $n$ nodes, $N$ edges, adjacency matrix $A$, and any $\epsilon, \delta > 0$, the number of triangles in $G$ can be computed with at most $\epsilon$ absolute error in a single pass over the edges using Algorithm~\ref{alg:general-ftc} in time $O((m \log m + N) \log 1/\delta)$ using $O(m \log 1/\delta)$ space with $m = \Omega(N^3/\epsilon^2)$ with probability at least $1 - \delta$.
\end{corollary}

\begin{proof}
    Let $\tX_1, \tX_2, \tX_3 = A$ and $E = \set{(1, 6), (2, 3), (4, 5)}$. Then the result follows directly from Theorem~\ref{thm:errorbound}. Note that there is an additional $N^3$ term in the bound on $m$, in order to remove the relative term $\norm{\tX}_{\frob} = \norm{A}_{\frob}^3 = N^{3/2}$ from Definition~\ref{def:approx-tensor-contraction}.  
\end{proof}

\section*{Acknowledgment}

We thank Ioannis Panageas for his valuable feedback.

\bibliographystyle{alpha}
\bibliography{references}

@String{Computing = "Computing" }

@String{Computer = "{IEEE} Computer" }

@String{Springer = "Springer-Verlag" }

@inproceedings{alon1999tracking,
  title={Tracking join and self-join sizes in limited storage},
  author={Alon, Noga and Gibbons, Phillip B and Matias, Yossi and Szegedy, Mario},
  booktitle={Proceedings of the eighteenth ACM SIGMOD-SIGACT-SIGART Symposium on Principles of Database Systems},
  pages={10--20},
  year={1999}
}

@inproceedings{dobra2002processing,
  title={Processing complex aggregate queries over data streams},
  author={Dobra, Alin and Garofalakis, Minos and Gehrke, Johannes and Rastogi, Rajeev},
  booktitle={Proceedings of the 2002 ACM SIGMOD International Conference on Management of Data},
  pages={61--72},
  year={2002}
}

@inproceedings{alon1996space,
  title={The space complexity of approximating the frequency moments},
  author={Alon, Noga and Matias, Yossi and Szegedy, Mario},
  booktitle={ACM Symposium on Theory of computing (STOC)},
  pages={20--29},
  year={1996}
}

@inproceedings{cormode2005sketching,
  title={Sketching streams through the net: Distributed approximate query tracking},
  author={Cormode, Graham and Garofalakis, Minos},
  booktitle={International Conference on Very large Data Bases (VLDB)},
  pages={13--24},
  year={2005}
}

@inproceedings{charikar2002finding,
  title={Finding frequent items in data streams},
  author={Charikar, Moses and Chen, Kevin and Farach-Colton, Martin},
  booktitle={International Colloquium on Automata, Languages, and Programming (ICALP)},
  pages={693--703},
  year={2002},
  organization={Springer}
}

@article{leis2015good,
  title={How good are query optimizers, really?},
  author={Leis, Viktor and Gubichev, Andrey and Mirchev, Atanas and Boncz, Peter and Kemper, Alfons and Neumann, Thomas},
  journal={Proceedings of the VLDB Endowment},
  volume={9},
  number={3},
  pages={204--215},
  year={2015},
  publisher={VLDB Endowment}
}

@article{pagh2013compressed,
  title={Compressed matrix multiplication},
  author={Pagh, Rasmus},
  journal={ACM Transactions on Computation Theory (TOCT)},
  volume={5},
  number={3},
  pages={1--17},
  year={2013},
  publisher={ACM New York, NY, USA}
}

@article{wang2015fast,
  title={Fast and guaranteed tensor decomposition via sketching},
  author={Wang, Yining and Tung, Hsiao-Yu and Smola, Alexander J and Anandkumar, Anima},
  journal={Advances in Neural Information Processing Systems (NeurIPS)},
  volume={28},
  year={2015}
}

@inproceedings{pham2013fast,
  title={Fast and scalable polynomial kernels via explicit feature maps},
  author={Pham, Ninh and Pagh, Rasmus},
  booktitle={ACM SIGKDD international conference on Knowledge discovery and data mining},
  pages={239--247},
  year={2013}
}

@inproceedings{weinberger2009feature,
  title={Feature hashing for large scale multitask learning},
  author={Weinberger, Kilian and Dasgupta, Anirban and Langford, John and Smola, Alex and Attenberg, Josh},
  booktitle={International Conference on Machine Learning (ICML)},
  pages={1113--1120},
  year={2009}
}

@inproceedings{heddes2024convolution,
  title={Convolution and Cross-Correlation of Count Sketches Enables Fast Cardinality Estimation of Multi-Join Queries},
  author={Heddes, Mike and Nunes, Igor  and Givargis, Tony and Nicolau, Alex},
  booktitle={Proceedings of the 2024 ACM SIGMOD International Conference on Management of Data},
  year={2024}
}

@inproceedings{ahle2020oblivious,
  title={Oblivious sketching of high-degree polynomial kernels},
  author={Ahle, Thomas D and Kapralov, Michael and Knudsen, Jakob BT and Pagh, Rasmus and Velingker, Ameya and Woodruff, David P and Zandieh, Amir},
  booktitle={ACM-SIAM Symposium on Discrete Algorithms (SODA)},
  pages={141--160},
  year={2020},
  organization={SIAM},
    note={See also the extended arXiv preprint arXiv:1909.01410}
}

@article{avron2014subspace,
  title={Subspace embeddings for the polynomial kernel},
  author={Avron, Haim and Nguyen, Huy and Woodruff, David},
  journal={Advances in Neural Information Processing Systems (NeurIPS)},
  volume={27},
  year={2014}
}

@book{golub2013matrix,
  title={Matrix computations},
  author={Golub, Gene H and Van Loan, Charles F},
  year={2013},
  publisher={JHU press}
}

@article{bridgeman2017hand,
  title={Hand-waving and interpretive dance: an introductory course on tensor networks},
  author={Bridgeman, Jacob C and Chubb, Christopher T},
  journal={Journal of physics A: Mathematical and theoretical},
  volume={50},
  number={22},
  pages={223001},
  year={2017},
  publisher={IOP Publishing}
}

@inproceedings{jowhari2005new,
  title={New streaming algorithms for counting triangles in graphs},
  author={Jowhari, Hossein and Ghodsi, Mohammad},
  booktitle={Computing and Combinatorics: 11th Annual International Conference, COCOON 2005 Kunming, China, August 16--19, 2005 Proceedings 11},
  pages={710--716},
  year={2005},
  organization={Springer}
}

@article{arad2010quantum,
  title={Quantum computation and the evaluation of tensor networks},
  author={Arad, Itai and Landau, Zeph},
  journal={SIAM Journal on Computing},
  volume={39},
  number={7},
  pages={3089--3121},
  year={2010},
  publisher={SIAM}
}

@article{henzinger1998computing,
  title={Computing on data streams.},
  author={Henzinger, Monika Rauch and Raghavan, Prabhakar and Rajagopalan, Sridhar},
  journal={External memory algorithms},
  volume={50},
  pages={107--118},
  year={1998}
}

@article{feigenbaum2005graph,
  title={On graph problems in a semi-streaming model},
  author={Feigenbaum, Joan and Kannan, Sampath and McGregor, Andrew and Suri, Siddharth and Zhang, Jian},
  journal={Theoretical Computer Science},
  volume={348},
  number={2-3},
  pages={207--216},
  year={2005},
  publisher={Elsevier}
}

@inproceedings{feigenbaum2004graph,
  title={On graph problems in a semi-streaming model},
  author={Feigenbaum, Joan and Kannan, Sampath and McGregor, Andrew and Suri, Siddharth and Zhang, Jian},
  booktitle={International Colloquium on Automata, Languages, and Programming (ICALP)},
  pages={531--543},
  year={2004},
  organization={Springer}
}

@inproceedings{feigenbaum2005graphSODA,
  title={Graph distances in the streaming model: the value of space.},
  author={Feigenbaum, Joan and Kannan, Sampath and McGregor, Andrew and Suri, Siddharth and Zhang, Jian},
  booktitle={ACM-SIAM Symposium on Discrete Algorithms (SODA)},
  volume={5},
  pages={745--754},
  year={2005}
}

@article{oseledets2011tensor,
  title={Tensor-train decomposition},
  author={Oseledets, Ivan V},
  journal={SIAM Journal on Scientific Computing},
  volume={33},
  number={5},
  pages={2295--2317},
  year={2011},
  publisher={SIAM}
}

@article{novikov2015tensorizing,
  title={Tensorizing neural networks},
  author={Novikov, Alexander and Podoprikhin, Dmitrii and Osokin, Anton and Vetrov, Dmitry P},
  journal={Advances in Neural Information Processing Systems (NeurIPS)},
  volume={28},
  year={2015}
}

@article{woodruff2014sketching,
  title={Sketching as a tool for numerical linear algebra},
  author={Woodruff, David P},
  journal={Foundations and Trends{\textregistered} in Theoretical Computer Science},
  volume={10},
  number={1--2},
  pages={1--157},
  year={2014},
  publisher={Now Publishers, Inc.}
}

@article{ibaraki1984optimal,
  title={On the optimal nesting order for computing n-relational joins},
  author={Ibaraki, Toshihide and Kameda, Tiko},
  journal={ACM Transactions on Database Systems (TODS)},
  volume={9},
  number={3},
  pages={482--502},
  year={1984},
  publisher={ACM New York, NY, USA}
}

@article{markov2008simulating,
  title={Simulating quantum computation by contracting tensor networks},
  author={Markov, Igor L and Shi, Yaoyun},
  journal={SIAM Journal on Computing},
  volume={38},
  number={3},
  pages={963--981},
  year={2008},
  publisher={SIAM}
}

@article{robeva2019duality,
  title={Duality of graphical models and tensor networks},
  author={Robeva, Elina and Seigal, Anna},
  journal={Information and Inference: A Journal of the IMA},
  volume={8},
  number={2},
  pages={273--288},
  year={2019},
  publisher={Oxford University Press}
}

@article{wainwright2008graphical,
  title={Graphical models, exponential families, and variational inference},
  author={Wainwright, Martin J and Jordan, Michael I and others},
  journal={Foundations and Trends{\textregistered} in Machine Learning},
  volume={1},
  number={1--2},
  pages={1--305},
  year={2008},
  publisher={Now Publishers, Inc.}
}

@phdthesis{deeds2025data,
  title={Data-Aware Complexity Analysis and Program Optimization},
  author={Deeds, Kyle},
  year={2025},
  school={University of Washington}
}

@InProceedings{mahankali2024near,
  author =	{Mahankali, Arvind V. and Woodruff, David P. and Zhang, Ziyu},
  title =	{{Near-Linear Time and Fixed-Parameter Tractable Algorithms for Tensor Decompositions}},
  booktitle =	{15th Innovations in Theoretical Computer Science Conference (ITCS 2024)},
  pages =	{79:1--79:23},
  series =	{Leibniz International Proceedings in Informatics (LIPIcs)},
  ISBN =	{978-3-95977-309-6},
  ISSN =	{1868-8969},
  year =	{2024},
  volume =	{287},
  editor =	{Guruswami, Venkatesan},
  publisher =	{Schloss Dagstuhl -- Leibniz-Zentrum f{\"u}r Informatik},
  address =	{Dagstuhl, Germany},
  URL =		{https://drops.dagstuhl.de/entities/document/10.4230/LIPIcs.ITCS.2024.79},
  URN =		{urn:nbn:de:0030-drops-196078},
  doi =		{10.4230/LIPIcs.ITCS.2024.79}
}

@inproceedings{ma2021fast,
 author = {Ma, Linjian and Solomonik, Edgar},
 booktitle = {Advances in Neural Information Processing Systems},
 editor = {M. Ranzato and A. Beygelzimer and Y. Dauphin and P.S. Liang and J. Wortman Vaughan},
 pages = {24299--24312},
 publisher = {Curran Associates, Inc.},
 title = {Fast and accurate randomized algorithms for low-rank tensor decompositions},
 url = {https://proceedings.neurips.cc/paper_files/paper/2021/file/cbef46321026d8404bc3216d4774c8a9-Paper.pdf},
 volume = {34},
 year = {2021}
}

@inproceedings{malik2018low,
 author = {Malik, Osman Asif and Becker, Stephen},
 booktitle = {Advances in Neural Information Processing Systems},
 editor = {S. Bengio and H. Wallach and H. Larochelle and K. Grauman and N. Cesa-Bianchi and R. Garnett},
 pages = {},
 publisher = {Curran Associates, Inc.},
 title = {Low-Rank Tucker Decomposition of Large Tensors Using TensorSketch},
 url = {https://proceedings.neurips.cc/paper_files/paper/2018/file/45a766fa266ea2ebeb6680fa139d2a3d-Paper.pdf},
 volume = {31},
 year = {2018}
}

@article{jayram2008one,
 author = {Jayram, T. S. and Kumar, Ravi and Sivakumar, D.},
 title = {The One-Way Communication Complexity of Hamming Distance},
 year = {2008},
 pages = {129--135},
 doi = {10.4086/toc.2008.v004a006},
 publisher = {Theory of Computing},
 journal = {Theory of Computing},
 volume = {4},
 number = {6},
 URL = {https://theoryofcomputing.org/articles/v004a006},
}

\appendix

\clearpage
\section{Proofs}
\label{sec:proofs}

This appendix contains proofs for the results presented in the paper. 

\subsection{Proof of Lemma~\ref{lemma:complement-count-sketch}}
\label{sec-proof-lemma-complement-count-sketch}

\begin{proof}
We first use the definition of the discrete Fourier transform and the count sketch to obtain
    \begin{align*}
        \rdbr{\overline{\dft C x}}_k &= \overline{\sum_{j=1}^m e^{-i 2\pi \frac{k-1}{m}(j-1)} \sum_{i=1}^n x_i s(i) \sqbr{h(i)=j}}\\
        &= \sum_{j=1}^m e^{i 2\pi \frac{k-1}{m}(j-1)} \sum_{i=1}^n x_i s(i) \sqbr{h(i)=j}.
\intertext{This uses the following properties of the complex conjugate: $\overline{x + y} = \overline{x} + \overline{y}$, $\overline{xy} = \overline{x} \, \overline{y}$, $\overline{\exp(i)} = \exp(-i)$, and $\overline{x} = x$ when $x \in \reals$. Using the definition of the complement count sketch, we have}
        \rdbr{\dft C' x}_k &= \sum_{l=1}^m e^{-i 2\pi \frac{k-1}{m}(l-1)} \sum_{i=1}^n x_i s(i) \sqbr{-h(i)+2 \equiv l \,(\mathrm{mod}\, m)}.\\
\intertext{Using a change of variable, setting $l = -j + 2$, we obtain}
        \rdbr{\dft C' x}_k &= \sum_{j=1}^m e^{-i 2\pi \frac{k-1}{m}(-j + 2 -1)} \sum_{i=1}^n x_i s(i) \sqbr{-h(i) + 2 \equiv -j + 2 \,(\mathrm{mod}\, m)}\\
        &= \sum_{j=1}^m e^{i 2\pi \frac{k-1}{m}(j-1)} \sum_{i=1}^n x_i s(i) \sqbr{h(i)=j} = \rdbr{\overline{\dft C x}}_k,
     \end{align*}
     which concludes the proof.
\end{proof}

\subsection{Proof of Lemma~\ref{lemma:general-expectation-variance}}
\label{sec:proof-lemma-general-expectation-variance}

\begin{proof}
First, note that every tensor can be written as a linear sum of the tensor product of standard basis vectors: $\tX_k = \sum_{\ivi_k = \ivone}^{\ivn_k} \etX_k(\ivi_k) \bigtimes_{u \in V(k)} e_{i_u}$, where $\bigtimes_{u \in S}$ denotes the tensor product for each member $u$ of $S$. Since $x_k = \rdbr*{\bm{\bullet}_{u \in V(k)} \dft C_u}\vecop(\tX_k)$, by linearity and the mixed-product identity $(A \bullet B)(x \otimes y) = (Ax) \circ (By)$,
\begin{align*}
    x_k = \rdbr*{\underset{u \in V(k)}{\bm{\bullet}} \dft C_u} \rdbr*{\sum_{\ivi_k = \ivone}^{\ivn_k} \etX_k(\ivi_k) \bigotimes_{u \in V(k)} e_{i_u}} = \sum_{\ivi_k = \ivone}^{\ivn_k} \etX_k(\ivi_k) \rdbr*{\underset{u \in V(k)}{\bm{\circ}} \dft C_u e_{i_u}},
\end{align*}
where $\bm{\circ}_{u \in S}$ is the Hadamard product for each member $u$ of $S$. This expresses the sketch $x_k$ of tensor $\mathcal{X}_k$ as a sum over its entries. Note that $\bm{\bullet}_{u \in V(k)} \dft C_u$ is the Fourier transform of a tensor sketch matrix and $\vecop(\tX_k) = \sum_{\ivi_k = \ivone}^{\ivn_k} \etX_k(\ivi_k) \bigotimes_{u \in V(k)} e_{i_u}$.
The estimate is then obtained by combining the sketches with the Hadamard product as follows: 
\begin{align}
    &\hat{\mathrm{tc}}(\mathcal{X}_1 \times \cdots \times \mathcal{X}_p, E) = e_1 \tran \idft (x_1 \circ \cdots \circ x_p)\\ 
    &= \sum_{\mathbf{i} = \bm{1}}^{\mathbf{n}} \mathcal{X}_1(\mathbf{i}_1) \cdots \mathcal{X}_p(\mathbf{i}_p) e_1 \tran \idft  \left( \underset{u \in [q]}{\bm{\circ}} F C_u e_{i_u} \right) \\
    &= \sum_{\mathbf{i} = \bm{1}}^{\mathbf{n}} \mathcal{X}_1(\mathbf{i}_1) \cdots \mathcal{X}_p(\mathbf{i}_p) e_1 \tran \idft  \left(\left( \underset{(u,v) \in E}{\bm{\circ}} F C_u e_{i_u}\right) \circ \left(\underset{(u,v) \in E}{\bm{\circ}}  F C_v e_{i_v} \right)\right) \label{eq:split}\\
    &= \sum_{\mathbf{i} = \bm{1}}^{\mathbf{n}} \mathcal{X}_1(\mathbf{i}_1) \cdots \mathcal{X}_p(\mathbf{i}_p) e_1 \tran \idft  \left(\left( \underset{(u,v) \in E}{\bm{\circ}} F C_u e_{i_u}\right) \circ \left(\underset{(u,v) \in E}{\bm{\circ}}  F C'_u e_{i_v} \right)\right)\\
    &= \sum_{\mathbf{i} = \bm{1}}^{\mathbf{n}} \mathcal{X}_1(\mathbf{i}_1) \cdots \mathcal{X}_p(\mathbf{i}_p) e_1 \tran \idft  \left( \left(\underset{(u,v) \in E}{\bm{\circ}} F C_u e_{i_u} \right) \circ \overline{\left(\underset{(u,v) \in E}{\bm{\circ}}  F C_u e_{i_v} \right)} \right) \label{eq:conjugated}\\
    &= \frac{1}{m} \sum_{\mathbf{i} = \bm{1}}^{\mathbf{n}} \mathcal{X}_1(\mathbf{i}_1) \cdots \mathcal{X}_p(\mathbf{i}_p) \left\langle \underset{(u,v) \in E}{\bm{\circ}} F C_u e_{i_u}, \underset{(u,v) \in E}{\bm{\circ}}  F C_u e_{i_v} \right\rangle \label{eq:inner}\\
    &= \frac{1}{m} \sum_{\mathbf{i} = \bm{1}}^{\mathbf{n}} \mathcal{X}_1(\mathbf{i}_1) \cdots \mathcal{X}_p(\mathbf{i}_p) \left\langle F T \bigotimes_{(u,v) \in E} e_{i_u}, F T \bigotimes_{(u,v) \in E} e_{i_v} \right\rangle \label{eq:tensor-sketch}\\
    &= \sum_{\mathbf{i} = \bm{1}}^{\mathbf{n}} \mathcal{X}_1(\mathbf{i}_1) \cdots \mathcal{X}_p(\mathbf{i}_p) \left\langle T \bigotimes_{(u,v) \in E} e_{i_u}, T \bigotimes_{(u,v) \in E} e_{i_v} \right\rangle \label{eq:normal-mat}.
\end{align}
Eq.~(\ref{eq:split}) divides the count sketches and the complement count sketches into two groups. Note that for every contraction, a count sketch is assigned to one mode and its complement to the other, so there is an equal number of them. Eq.~(\ref{eq:conjugated}) uses Lemma~\ref{lemma:complement-count-sketch} and Eq.~(\ref{eq:inner}) uses $e_1^{T}F^{-1}(x \circ \overline{y}) = \frac{1}{m}\langle x, y \rangle$. This step of the proof requires each contraction to have one conjugated mode which is achieved using the complement count sketches. Eq.~(\ref{eq:tensor-sketch}) uses the fact that Eq.~(\ref{eq:inner}) includes a Fourier transform of a tensor sketch. Eq.~(\ref{eq:normal-mat}) uses the invariance of the inner product to a unitary transformation, given that $F$ is a unitary matrix times $\sqrt{m}$.
The expectation of the estimator is given by
\begin{gather}
    \mathbb{E}[\hat{\mathrm{tc}}(\mathcal{X}_1 \times \cdots \times \mathcal{X}_p, E)] = \sum_{\mathbf{i} = \bm{1}}^{\mathbf{n}} \mathcal{X}_1(\mathbf{i}_1) \cdots \mathcal{X}_p(\mathbf{i}_p) \mathbb{E}\left\langle T \bigotimes_{(u,v) \in E} e_{i_u}, T \bigotimes_{(u,v) \in E} e_{i_v} \right\rangle\\
    \mathbb{E}\left\langle T \bigotimes_{(u,v) \in E} \mkern-5mu e_{i_u}, T \bigotimes_{(u,v) \in E} \mkern-5mu e_{i_v} \right\rangle = \left\langle \bigotimes_{(u,v) \in E} \mkern-5mu e_{i_u}, \bigotimes_{(u,v) \in E} \mkern-5mu e_{i_v} \right\rangle = \prod_{(u,v) \in E} \mkern-9mu \left\langle  e_{i_u},  e_{i_v} \right\rangle = \prod_{(u,v) \in E} \mkern-9mu [i_u = i_v] \label{eq:iverson-brackets}.
\end{gather}
The expected value in Eq.~\ref{eq:iverson-brackets} and the variance bound follow directly from Lemma~\ref{lemma:tensor-sketch}. The squared Frobenius norms for each tensor in the variance bound stem from the sum of tensor components. Note that in our proof, we do not assume anything about the structure of the tensor network. 
\end{proof}

\subsection{Proof of Lemma~\ref{lemma:exponential-dependence}}
\label{sec:proof-lemma-exponential-dependence}

\begin{proof}
    Since $\vecop(\tX_i) = \vone_n \otimes \vone_n$ for $2 \leq i \leq q$, we use the convolution theorem $\dft(x \star y) = \overline{\dft x} \circ \dft y$ and the mixed-product identity $(A \bullet B)(x \otimes y) = (Ax) \circ (By)$ to obtain
    \begin{align*}
        y = e_1\tran\idft((\dft C_1 \vone) \circ \overline{(\dft C_1\vone) \circ (\dft C_2\vone)} \circ \cdots \circ \overline{(\dft C_{q-1}\vone) \circ (\dft C_q\vone)} \circ (\dft C_{q} \vone)),
    \end{align*}
    where the complex conjugate is taken over every other pair of count sketches. We can use the associativity of the Hadamard product to rearrange the expression so that all the conjugated terms are on one side. Note that $e_1\tran\idft(\overline{x} \circ y) = \frac{1}{m}\vone\tran(\overline{x} \circ y) = \frac{1}{m}\anbr{x, y}$, and $\norm{x}^2_2 = \anbr{x, x}$, giving
    \begin{align*}
        \frac{1}{m}\anbr{(\dft C_1\vone) \circ \cdots \circ (\dft C_q\vone), (\dft C_1\vone) \circ \cdots \circ (\dft C_q\vone)}
        = \frac{1}{m}\norm{(\dft C_1\vone) \circ \cdots \circ (\dft C_q\vone)}_2^2.
    \end{align*}
    We now apply \citep[Lemma 43]{ahle2020oblivious}, since $\norm{T(x^{\otimes q})}_2^2 = \frac{1}{m}\norm{(\dft C_1\vone) \circ  \cdots \circ (\dft C_q\vone)}_2^2$ because $\idft$ is a unitary matrix times $\frac{1}{\sqrt{m}}$. Lastly, note that $\norm{x^{\otimes q}}_2^4 = \prod_{i=1}^{q + 1} \norm{\tX_i}_{\frob}^2$, concluding the proof.
\end{proof}

\subsection{Proof of Lemma~\ref{lemma:recursive-AFTNC}}
\label{sec:proof-lemma-recursive-AFTNC}

\begin{proof}
    We can assume that the TNC satisfies the WLOG assumptions from Section~\ref{sec:wlog}.
    Then, by expanding $g(o)$ and the definition of the inner product, we have
    \begin{align*}
         \vecop(\tX_o) \tran g(o) &= \sum_{\ivi_o = \ivone}^{\ivn} \tX_o(\ivi_o) \rdbr*{\bigotimes_{l_1 \in \Gamma(o)} \mat(\tX_{l_1})g(l_1)}(\ivi_o),
     \end{align*}
     where $\ivi_k$ contains all the indices of tensor $\tX_k$.
     Let $\ivo_k$ contain all the indices associated with outgoing edges of tensor $\tX_k$, and let $\ivl_i$ contain all the indices of all the tensors at a distance $i$ from the root, also called the $i$th level of the tree. Note that $\ivi_o = \ivl_1$ since the indices of the root tensor are all the indices in the first level. Let $E(l)$ denote the subset of contractions that are part of the first $l$ levels of the tree. We then obtain the following:
     \begin{align*}    
         \vecop(\tX_o) \tran g(o) &= \sum_{\ivl_1 = \ivone}^{\ivn} \sum_{\ivl_2 = \ivone}^{\ivn} \tX_o(\ivi_o) \rdbr*{\prod_{l_1 \in \Gamma(o)} \tX_{l_1}(\ivi_{l_1}) g(l_1)(\ivo_{l_1})} \prod_{(u, v) \in E(1)} \sqbr{i_u = i_v},
    \end{align*}
    where the indices in $\ivi_o$ are doubled with the introduction of $\ivi_{l_1}$. The added Iverson brackets ensures that the two indices are equal.
    We can then repeat this expansion of the recursion for all the $w$ levels of the tree, giving us
    \begin{align*}
         \vecop(\tX_o) \tran g(o) &= \sum_{\ivl_1 = \ivone}^{\ivn} \sum_{\ivl_2 = \ivone}^{\ivn} \cdots \sum_{\ivl_{w} = \ivone}^{\ivn} \tX_o(\ivi_o) \rdbr*{\prod_{l_1 \in \Gamma(o)} \tX_{l_1}(\ivi_{l_1}) \rdbr*{\prod_{l_2 \in \Gamma(l_1)} \tX_{l_2}(\ivi_{l_2}) \dots }} \prod_{(u, v) \in E} \sqbr*{i_u = i_v}\\
         &= \sum_{\ivi = \ivone}^{\ivn} \tX_1(\ivi_1) \cdots \tX_p(\ivi_p) \prod_{(u,v) \in E} \sqbr*{i_u = i_v},
    \end{align*}
    where $\ivi$ contains all the indices since this is a full TNC, concluding the proof.
\end{proof}

\subsection{Proof of Lemma~\ref{lemma:acyclic-expectation-variance}}
\label{sec:proof-lemma-acyclic-expectation-variance}

\begin{proof}
    The proof proceeds by first establishing the expectation, followed by the variance bound.
    \paragraph{Expectation.} Note that all the recursive sketch matrices $R_1, \dots, R_p$ are independent from each other. Then, by the law of total expectation, Lemma~\ref{lemma:recursive-sketch}, and with $x_o = \vecop(\tX_o)$,
    \begin{align*}
        \E\sqbr{y} = \E\sqbr*{x_o\tran R_o\tran R_o \hat{g}(o)} = \E\sqbr*{\E\sqbr*{x_o\tran R_o\tran R_o \hat{g}(o) \middle| \hat{g}(o)} }
        = x_o\tran \rdbr*{\bigotimes_{l \in \Gamma(o)} \E\sqbr*{\mat(\tX_l)R_l\tran R_l \hat{g}(l)}}.
    \end{align*}
    It follows that $\E\sqbr{ x_o\tran R_o\tran R_o \hat{g}(o)} = x_o \tran g(o)$ since $\E\sqbr*{\mat(\tX_k)R_k\tran R_k \hat{g}(k) \middle| \hat{g}(k)} = \mat(\tX_k)\hat{g}(k)$ for any $k$, which matches the result in Lemma~\ref{lemma:recursive-AFTNC}.

    \paragraph{Variance bound.} By the law of total variance, Lemma~\ref{lemma:recursive-sketch}, and with $x_o = \vecop(\tX_o)$,
\begin{align*}
    \Var\rdbr*{\hat{\tc}(\tX_1 \times \cdots \times \tX_p, E)} &= \E\sqbr*{\Var\rdbr*{x_o\tran R_o\tran R_o \hat{g}(o) \middle| \hat{g}(o)}} + \Var\rdbr*{\E\sqbr*{x_o\tran R_o\tran R_o \hat{g}(o) \middle| \hat{g}(o)}}\\
    &= \rdbr*{\!\rdbr*{1 + \frac{8}{m}}^{2q_o} - 1}\norm{\tX_o}^2_{\frob}\E\sqbr*{\norm{\hat{g}(o)}^2_2} + \Var\rdbr*{x_o\tran \hat{g}(o)}.
\end{align*}
For the norm of $\hat{g}$, with $M_l = \mat(\tX_l)$, we have that
\begin{align*}
    \E\sqbr*{\norm{\hat{g}(k)}^2_2 \middle| \hat{g}(l) \; \forall l \in \Gamma(k)} &= \prod_{l \in \Gamma(k)} \E\sqbr*{\norm{\mat(\tX_l) R_l\tran R_l \hat{g}(l)}^2_2 \middle| \hat{g}(l)}\\
    &= \prod_{l \in \Gamma(k)} \sum_{i=1}^n \E\sqbr*{\rdbr*{M_l(i, :) R_l\tran R_l \hat{g}(l)}^2 \middle| \hat{g}(l)}\\
    &= \prod_{l \in \Gamma(k)} \sum_{i=1}^n \Var\rdbr*{M_l(i, :) R_l\tran R_l \hat{g}(l) \middle| \hat{g}(l)} + \E\sqbr*{M_l(i, :) R_l\tran R_l \hat{g}(l) \middle| \hat{g}(l)}^2.
\intertext{Using Lemma~\ref{lemma:recursive-sketch} and the Cauchy-Schwarz inequality we get the following bound:}
    \E\sqbr*{\norm{\hat{g}(k)}^2_2 \middle| \hat{g}(l) \; \forall l \in \Gamma(k)} &\leq \prod_{l \in \Gamma(k)} \rdbr*{\!\rdbr*{1 + \frac{8}{m}}^{2q_l - 2} - 1} \norm{\tX_l}^2_{\frob} \E\sqbr*{\norm{\hat{g}(l)}^2_2} + \norm{\tX_l}^2_{\frob} \E\sqbr*{\norm{\hat{g}(l)}^2_2}\\
    &= \prod_{l \in \Gamma(k)} \rdbr*{1 + \frac{8}{m}}^{2q_l - 2} \norm{\tX_l}^2_{\frob} \E\sqbr*{\norm{\hat{g}(l)}^2_2}.
\end{align*}
Going back to the remaining variance term, we have that 
\begin{align*}
    \Var\rdbr*{x_o\tran \hat{g}(o)} &= \E\sqbr*{\rdbr*{x_o\tran \rdbr*{{\textstyle\bigotimes_{l \in \Gamma(o)}} \mat(\tX_l) R_l\tran R_l \hat{g}(l)} - x_o\tran \rdbr*{{\textstyle\bigotimes_{l \in \Gamma(o)}} \mat(\tX_l) g(l)}}^2}\\
    &= \E\sqbr*{\rdbr*{x_o\tran \rdbr*{{\textstyle\bigotimes_{l \in \Gamma(o)}} \mat(\tX_l) R_l\tran R_l \hat{g}(l) - \mat(\tX_l) g(l)}}^2 }\\
    &\leq \norm{\tX_o}^2_{\frob} \, \prod_{l \in \Gamma(o)} \E\sqbr*{\norm{ \mat(\tX_l) R_l\tran R_l \hat{g}(l) - \mat(\tX_l) g(l)}_2^2 } \eqnote{(Cauchy-Schwarz)}\\
    &= \norm{\tX_o}^2_{\frob} \, \prod_{l \in \Gamma(o)} \sum_{i=1}^n \E\sqbr*{\rdbr*{ M_l(i, :) R_l\tran R_l \hat{g}(l) - M_l(i, :) g(l)}^2 }\\
    &= \norm{\tX_o}^2_{\frob} \, \prod_{l \in \Gamma(o)} \sum_{i=1}^n \Var\rdbr*{M_l(i, :) R_l\tran R_l \hat{g}(l)}\\
\intertext{We then apply the law of total variance to obtain:}
    \Var\rdbr*{x_o\tran \hat{g}(o)} &\leq \norm{\tX_o}^2_{\frob} \, \prod_{l \in \Gamma(o)} \sum_{i=1}^n \E\sqbr*{\Var\rdbr*{M_l(i, :) R_l\tran R_l \hat{g}(l) \middle| \hat{g}(l)}} + \Var\rdbr*{\E\sqbr*{M_l(i, :) R_l\tran R_l \hat{g}(l) \middle| \hat{g}(l)}}\\
    &\leq \norm{\tX_o}^2_{\frob} \, \prod_{l \in \Gamma(o)} \rdbr*{\rdbr*{1 + \frac{8}{m}}^{2q_l - 2} - 1} \norm{\tX_l}^2_{\frob}\E\sqbr*{\norm{\hat{g}(l)}^2_2} + \sum_{i=1}^n \Var\rdbr*{M_l(i, :) \hat{g}(l)}.
\end{align*}
Combining the results thus far gives us
\begin{align*}
    \Var\rdbr*{y}
    &\leq \rdbr*{\rdbr*{1 + \frac{8}{m}}^{2q_o} - 1}\norm{\tX_o}^2_{\frob} \, \prod_{l \in \Gamma(o)} \rdbr*{1 + \frac{8}{m}}^{2q_l - 2}\norm{\tX_l}^2_{\frob} \E\sqbr*{\norm{\hat{g}(l)}^2_2}\\ &+ \norm{\tX_o}^2_{\frob} \, \prod_{l \in \Gamma(o)} \rdbr*{\rdbr*{1 + \frac{8}{m}}^{2q_l - 2} - 1} \norm{\tX_l}^2_{\frob}\E\sqbr*{\norm{\hat{g}(l)}^2_2} + \sum_{i=1}^n \Var\rdbr*{M_l(i, :) \hat{g}(l)},
\end{align*}
The first term recursively aggregates all the tensor norms and $(1 + 8/m)$ coefficients by traversing the tensor diagram from the root to the leaves. In total, there are $q_o + \sum_{k=1 : k \neq o}^p q_k - 1 = q - p + 1 = t$ of such $(1 + 8/m)$ coefficients. Note that, for each leaf of the tree, $\hat{g}(l) = 1$, which has zero variance. If we set $t' = t - q_o$, then this results in the following:
\begin{align*}
     \Var\rdbr*{y} &\leq \rdbr*{1 + \frac{8}{m}}^{2t} \prod_{k=1}^p \norm{\tX_k}^2_{\frob} - \rdbr*{1 + \frac{8}{m}}^{2t'} \prod_{k=1}^p \norm{\tX_k}^2_{\frob}\\ &+ \norm{\tX_o}^2_{\frob} \, \prod_{l \in \Gamma(o)} \rdbr*{\!\rdbr*{1 + \frac{8}{m}}^{2q_l - 2} - 1}\norm{\tX_l}^2_{\frob}\E\sqbr*{\norm{\hat{g}(l)}^2_2}  + \sum_{i=1}^n \Var\rdbr*{M_l(i, :) \hat{g}(l)},
\end{align*}
The final inequality is derived from the observation that the overcounting of the norms of the leaf tensors results in the cancellation of the last two terms. However, it should be noted that we are overcounting by no less than the product of the norms, thus ultimately yielding
\begin{align*}
     \Var\rdbr*{y} &\leq  \rdbr*{1 + \frac{8}{m}}^{2t} \prod_{k=1}^p \norm{\tX_k}^2_{\frob} - \rdbr*{1 + \frac{8}{m}}^{2t'} \prod_{k=1}^p \norm{\tX_k}^2_{\frob} + \rdbr*{1 + \frac{8}{m}}^{2t'} \prod_{k=1}^p \norm{\tX_k}^2_{\frob} - \prod_{k=1}^p \norm{\tX_k}^2_{\frob},
\end{align*}
which equals to $\rdbr*{\!\rdbr*{1 + \frac{8}{m}}^{2t} - 1} \prod_{k=1}^p \norm{\tX_k}^2_{\frob}$, concluding the proof.
\end{proof}

\subsection{Proof of Theorem~\ref{thm:errorbound}}
\label{sec:proof-thm-errorbound}

\begin{proof} 
Let $y = \hat{\tc}(\tX_1 \times \cdots \times \tX_p, E)$, then, since this is a full TNC,
\begin{align*}
    \E\sqbr*{\norm*{y - \tc(\tX_1 \times \cdots \times \tX_p, E)}^2_{\frob}}
    = \E\sqbr*{\rdbr*{y - \tc(\tX_1 \times \cdots \times \tX_p, E)}^2} = \Var\rdbr*{y}.
\end{align*}
Note that $\rdbr*{1 + 8/m}^{2t} - 1 \leq \exp\rdbr*{16t/m} - 1 \leq 32t/m$ when $m \geq 16t$.
Therefore, by {\bf (1)} Lemma~\ref{lemma:acyclic-expectation-variance}, let $m = \Omega(t/ \epsilon^{2})$; and {\bf (2)} Lemma~\ref{lemma:general-expectation-variance}, let $m = \Omega(3^t / \epsilon^2)$ to obtain an estimator for full ATNC (see Definition~\ref{def:approx-tensor-contraction}) with constant success probability using the Chebyshev inequality. 
Note that $\prod_{k=1}^{p} \norm{\tX_k}_{\frob} = \norm{\tX_1 \times \cdots \times \tX_p}_{\frob}$.
We obtain success with high probability by selecting the median of $O(\log{1/\delta})$ independent and identically distributed estimates with constant success probability using the Chernoff bound.
\end{proof}

\subsection{Proof of Corollary~\ref{corollary:full-to-partial-tc}}
\label{sec:proof-corollary-full-to-partial-tc}

\begin{proof}
    Expanding the squared Frobenius norm of the error, with $\hat{\tY}$ and $\tY$ as the approximate and exact TNC, respectively, gives
\begin{align*}
    \E\sqbr*{\norm*{\hat{\tc}(\tX_1 \times \cdots \times \tX_p, E) - \tc(\tX_1 \times \cdots \times \tX_p, E)}^2_{\frob}}
    = \sum_{\ivr = \ivone}^{\ivn} \E\sqbr*{\rdbr*{\hat{\tY}(\ivr) - \tY(\ivr)}^2} = \sum_{\ivr = \ivone}^{\ivn} \Var\rdbr*{\hat{\tY}(\ivr)}.
\end{align*}
Let $\tX_k(\ivs_k)$ denote the subtensor of $\tX_k$ used to compute the entry $\hat{\tY}(\ivr)$, with 
\begin{align*}
    \ivs_k = (\gamma(u) : u \in V(k)), \qquad \text{and} \qquad \gamma(u) = {\begin{cases}:&{\text{if }}u \in K,\\i_u&{\text{otherwise}},\end{cases}}
\end{align*}
where $V(k)$ are the modes of $\tX_k$ and $K$ the contracted modes.
Then, by {\bf (1)} Lemma~\ref{lemma:acyclic-expectation-variance}, let $m \geq 32t/ (\epsilon^{2}\delta)$; and {\bf (2)} Lemma~\ref{lemma:general-expectation-variance}, let $m \geq 3^t / (\epsilon^2 \delta)$, 
\begin{align*}
    \sum_{\ivr = \ivone}^{\ivn} \Var\rdbr*{\hat{\tY}(\ivr)} \leq \sum_{\ivr = \ivone}^{\ivn} \epsilon^2 \delta \prod_{k = 1}^p \norm{\tX_k(\ivs_k)}_{\frob}^2= \epsilon^2 \delta \prod_{k = 1}^p \norm{\tX_k}_{\frob}^2.
\end{align*}
By the Chebyshev inequality and Definition~\ref{def:approx-tensor-contraction}, this concludes the proof.
\end{proof}

\clearpage
\section{Additional algorithms}
\label{sec:additional-algorithms}

This appendix contains pseudocode for functions used by the main algorithms presented in the paper. We start with Algorithm~\ref{alg:tensor-sketch-hash}, which computes the sign and row of the tensor sketch matrix for a given column (which corresponds to a tensor index). This function is used in Algorithm~\ref{alg:recursive-sketch-hash}, which computes the sign and row of the recursive sketch matrix for a given column. This function, in turn, is used to simultaneously compute the sketch of a tensor and its matrix-vector product in Algorithm~\ref{alg:sketch-matrix-vector-product}.

The tensor sketch matrix of order $q$ is composed of $q$ count sketch matrices, which in turn consist of $q$ sign and row hash functions.
An alternative definition of the tensor sketch matrix $T\in \reals^{m \times n^q}$, specifies each component as $T(j, i) = s_1(i_1)\cdots s_q(i_q) \sqbr{(h_1(i_1) + \dots + h_q(i_q) - q + 1) \equiv j \pmod{m}}$, where $i = 1 + \sum_{k=1}^q (i_k - 1)n^{q-k}$ (the vectorization of tensor index $\ivi$). Based on this formulation, Algorithm~\ref{alg:tensor-sketch-hash} computes the sign and row of the tensor sketch matrix for a given column. This is achieved by multiplying the signs and summing the rows (modulo the sketch size), obtained by the count sketch hash functions for each mode.

\begin{algorithm}[h]
\caption{Tensor sketch matrix sign and row by column.}
\label{alg:tensor-sketch-hash}
\begin{algorithmic}[1]
\Require index $\ivi = \rdbr{i_1, \dots, i_q}$ and tensor sketch matrix $T \in \reals^{m \times n^q}$, which is composed of independent count sketch matrices $C_1, \dots, C_q \in \reals^{m \times n}$ that consist of sign and row hash functions $s_k$ and $h_k$ for all $k \in \sqbr{q}$
\Ensure $\sigma \in \set{-1, 1}$ and $b \in \sqbr{m}$
\Function{TensorSketchHash}{$\ivi, T$}
\State $\sigma \gets 1$
\State $b \gets 0$
\For{$k \in \sqbr{q}$}
\State $\sigma \gets \sigma \cdot s_k(i_k)$
\State $b \gets b + h_k(i_k)$
\EndFor
\State $b \gets ((b - q) \bmod m) + 1$
\State \Return $\sigma, b$
\EndFunction
\end{algorithmic}
\end{algorithm}

Algorithm~\ref{alg:recursive-sketch-hash} computes the sign and row of the recursive sketch matrix for a given column. This is achieved by reducing a list of row indices, which are obtained by the row hash functions of the count sketch matrices for each mode, in a binary tree fashion using Algorithm~\ref{alg:tensor-sketch-hash}. At each step, the signs are accumulated through multiplication.

\begin{algorithm}[h]
\caption{Recursive sketch matrix sign and row by column.}
\label{alg:recursive-sketch-hash}
\begin{algorithmic}[1]
\Require index $\ivi = \rdbr{i_1, \dots, i_q}$ and recursive sketch matrix $R \in \reals^{m \times n^q}$, which is composed of independent count sketch matrices $C_1, \dots, C_q \in \reals^{m \times n}$ that consist of sign and row hash functions $s_k$ and $h_k$ for all $k \in \sqbr{q}$, and independent tensor sketch matrices $T_{l,1}, \dots, T_{l, l/2} \in \reals^{m\times m^2}$ for $l \in \set{2, 4, \dots, q/2, q}$
\Ensure $\sigma \in \set{-1, 1}$ and $b \in \sqbr{m}$
\Function{RecursiveSketchHash}{$\ivi, R$}
\State $\sigma \gets 1$
\State $\vb \in \set{0}^q$
\For{$k \in \sqbr{q}$}
\State $\sigma \gets \sigma \cdot s_k(i_k)$
\State $\evb_k \gets h_k(i_k)$
\EndFor
\For{$l \in \set{q, q/2, \dots, 4, 2}$}
\For{$k \in \sqbr{l/2}$}
\State $s, h \gets \textsc{TensorSketchHash}(\rdbr{b_{2k-1},b_{2k}}, T_{l,k})$
\State $\sigma \gets \sigma \cdot s$
\State $\evb_k \gets h$
\EndFor
\EndFor
\State \Return $\sigma, b_1$
\EndFunction
\end{algorithmic}
\end{algorithm}

Algorithm~\ref{alg:sketch-matrix-vector-product} computes the size-$m$ sketch of an order-$q$ tensor $\tX$ while simultaneously multiplying the resulting sketch matrix with a vector. This is an important function for our second method as it significantly improves upon the complexity of first computing the sketch followed by the matrix-vector product, which would take time $O(q \nnz(\tX) + m^2)$ and space $O(m^2)$, and improves it to just $O(q \nnz(\tX))$ time and $O(m)$ space.

\begin{algorithm}[h]
\caption{Combined tensor sketching and matrix-vector product.}
\label{alg:sketch-matrix-vector-product}
\begin{algorithmic}[1]
\Require order $q$ tensor $\tX \in \reals^{n \times \cdots \times n}$, count sketch matrix $C \in \reals^{m \times n}$ that consists of sign and row hash functions $s$ and $h$, recursive sketch matrix $R \in \reals^{m \times n^{q - 1}}$, and any vector $z \in \reals^m$
\Ensure $y = C \mat(\tX) R\tran z$
\Function{SketchedMatrixVectorProduct}{$C, \tX, R, z$}
\State $\vy \in \set{0}^m$
\For{each component $\tX(\ivi)$}
\State $j \gets h(i_1)$
\State $\sigma, b = \textsc{RecursiveSketchHash}(\rdbr{i_2, i_3, \dots, i_q}, R)$
\State $\evy_{j} \gets \evy_{j} + \etX(\ivi) \cdot z_b \cdot \sigma \cdot s(i_1)$
\EndFor
\State \Return $y$
\EndFunction
\end{algorithmic}
\end{algorithm}

\end{document}